\documentclass{lmcs} 

\keywords{V-Formation, Model Predictive Control, Markov Decision Processes, Controller-Attacker Games}

\usepackage{todonotes}
\usepackage{url}
\usepackage{glossaries}

\usepackage[caption=false,font=footnotesize]{subfig}

\usepackage{cite}
\usepackage[utf8]{inputenc}
\usepackage{stfloats}

\usepackage{tikz}
\usepackage{multirow}
\usepackage{makecell}
\usepackage{tikz}
\usepackage{wasysym}
\usepackage{float}

\usepackage{amsmath,amssymb,amsfonts}
\usepackage{graphicx}
\usepackage[ruled,linesnumbered,lined,boxed,commentsnumbered]{algorithm2e}
\usepackage{booktabs}

\usepackage{xargs}                      

\newcommand{\xv}{{\boldsymbol {x}}}
\newcommand{\vv}{{\boldsymbol {v}}}
\newcommand{\va}{{\boldsymbol {a}}}
\newcommand{\vs}{{\boldsymbol {s}}}
\newcommand{\vd}{{\boldsymbol {d}}}
\newcommand{\VM}{{\it VM}}
\newcommand{\CV}{{\it CV}}
\newcommand{\UB}{{\it UB}}
 \newcommand{\ampc}{\mathtt{AMPC}}
 \newcommand{\dampc}{\mathtt{DAMPC}} 
 \newcommand{\lampc}{\mathtt{LocalAMPC}}
\newcommand{\M}{{\mathcal{M}}}
\newcommand{\PD}{{\it PD}}
\newcommand{\comment}[1]{{}}

\DeclareMathOperator*{\argmin}{\arg\min}

\newtheorem{definition}{Definition}
\newtheorem{proposition}{Proposition}
\newtheorem{corollary}{Corollary}
\newtheorem{remark}{Remark}
\theoremstyle{plain}\newtheorem{lemma}[thm]{Lemma} 

\makeglossaries
\newglossaryentry{pso}
{
    name=PSO,
    description={PSO}
}

\newglossaryentry{isp}
{
    name=isp,
    description={isp}
}
  
\newglossaryentry{mpc}
{
    name=MPC,
    description={MPC}
} 

\newglossaryentry{mdp}
{
    name=MDP,
    description={MDP}
} 

\begin{document}

\title{V-Formation as Model Predictive Control}

\author[R.~Grosu]{Radu Grosu}	
\address{Cyber-Physical Systems Group, Technische Universit\"at Wien, Austria}	
\email{radu.grosu@tuwien.ac.at}  

\author[A.~Lukina]{Anna Lukina}	
\address{Institute of Science and Technology Austria}	
\email{anna.lukina@ist.ac.at}  

\author[S.A.~Smolka]{Scott A. Smolka}	
\address{Department of Computer Science, Stony Brook University, Stony Brook, NY, USA}	
\email{sas@cs.stonybrook.edu}  

\author[A.~Tiwari]{Ashish Tiwari}	
\address{Microsoft Research, USA}	
\email{Ashish.Tiwari@microsoft.com}  

\author[V.~Varadarajan]{Vasudha Varadarajan}	
\address{Department of Computer Science, Stony Brook University, Stony Brook, NY, USA}	
\email{vvaradarajan@cs.stonybrook.edu}  

\author[X.~Wang]{Xingfang Wang}	
\address{Department of Computer Science, Stony Brook University, Stony Brook, NY, USA}	
\email{wxingfang@cs.stonybrook.edu}  





\begin{abstract}

We present recent results that demonstrate the power of viewing the problem of V-formation in a flock of birds as one of Model Predictive Control (MPC). The V-formation-MPC marriage can be understood in terms of the problem of synthesizing an optimal plan for a continuous-space and continuous-time Markov decision process (MDP), where the goal is to reach a target state that minimizes a given cost function.

The first result we consider is ARES, an efficient approximation algorithm for 
generating optimal plans (action sequences) that take an 
initial state of an \gls{mdp} to a state 
whose cost is below a specified (convergence) threshold.  ARES uses Particle Swarm Optimization, with \emph{adaptive sizing} for 
both the receding horizon and the particle swarm.  Inspired 
by Importance Splitting, the length of the horizon and the number of 
particles are chosen such that at least one particle reaches a 
\emph{next-level} state, i.e., a state where the cost decreases by a 
required delta from the previous-level state.  The level relation on states and the plans constructed by ARES implicitly define a Lyapunov function and an optimal policy, respectively, both of which could be explicitly generated by applying ARES to all states of the \gls{mdp}, up to some topological equivalence relation.
We assess the effectiveness of ARES by statistically evaluating its rate of success in generating optimal plans for V-formation. 
  
ARES can alternatively be viewed as a model-predictive control (MPC) algorithm that utilizes an adaptive receding horizon, a technique we refer to as Adaptive MPC (AMPC).  We next present Distributed AMPC (DAMPC), a distributed version of AMPC that works with local neighborhoods. We introduce adaptive neighborhood resizing, whereby the neighborhood size is determined by the cost-based Lyapunov function evaluated over a global system state.  Our approach applies to reachability problems for any collection of entities that seek convergence from an arbitrary initial state to a desired goal state, where a notion of distance to the goal state(s) can be suitably defined. Our experimental evaluation shows that DAMPC can perform almost as well as centralized AMPC, while using only local information and a form of distributed consensus in each time step.
  
Finally, inspired by security attacks on cyber-physical systems (CPS), we introduce \emph{controller-attacker games}, where two players, a controller and an attacker, have antagonistic objectives. To highlight the power of adaptation, we formulate  a special case of controller-attacker games called V-formation games, where the attacker's goal is to prevent the controller from attaining V-formation. We demonstrate how adaptation in the design of the controller helps in overcoming certain attacks.
  
    
\end{abstract}

\maketitle

\section{Introduction}
 Cyber-physical systems (CPSs) comprised of multiple computing agents are often highly distributed and may exhibit emergent behavior.
V-formation in a flock of birds is a quintessential example of emergent behavior in a (stochastic) multi-agent system. V-formation brings numerous benefits to the flock. It is primarily known for being energy-efficient due to the \emph{upwash benefit} a bird in the flock enjoys from its frontal neighbor. It also offers a \emph{clear view} benefit, as no bird's field of vision is obstructed by another bird in the formation. Moreover, its collective spatial flock mass can be intimidating to potential predators.  It is therefore not surprising that interest in V-formation is on the rise~\cite{droneswarm,Bloomberg}.
Because of V-formation's intrinsic appeal, it is important to (i)~understand its control-theoretic foundations, (ii)~devise efficient algorithms for the problem, and (iii)~identify the vulnerabilities in these approaches to cyber-attacks.

This paper brings together our recent results on V-formation that show how the problem can be formulated in terms of Model Predictive Control (MPC), both centralized and distributed.  It also shows how an MPC-based formulation of V-formation can be used as a comprehensive framework for investigating cyber-attacks on this formation.  

We first consider \emph{Adaptive Receding-Horizon Synthesis of Optimal Plans} (ARES) \cite{lukina-tacas17}, an efficient approximation algorithm for generating optimal plans (action sequences) that take an initial state of an MDP to a state whose cost is below a specified (convergence) threshold. ARES uses Particle Swarm Optimization (PSO), with \emph{adaptive sizing} for both the receding horizon and the particle swarm.  Inspired by Importance Splitting, a sampling technique for rare events, the length of the horizon and the number of particles are chosen such that at least one particle reaches a \emph{next-level} state, that is, a state where the cost decreases by a required delta from the previous-level state. The level relation on states and the plans constructed by ARES implicitly define a Lyapunov function and an optimal policy, respectively, both of which could be explicitly generated by applying ARES to all states of the MDP, up to some topological equivalence relation.

We assess the effectiveness of ARES by statistically evaluating its rate of success in generating optimal plans
that bring a flock from an arbitrary initial state to a state exhibiting a single connected V-formation.  For flocks with 7 birds, ARES is able to generate a plan that leads to a V-formation in 95\% of the 8,000 random initial configurations within 63 seconds, on average. ARES can be viewed as a model-predictive controller (MPC) with an adaptive receding horizon, which we also call adaptive MPC (AMPC). We provide statistical guarantees of convergence. To the best of our knowledge, our adaptive-sizing approach is the first to provide \emph{convergence guarantees} in receding-horizon 
techniques. 

We next present DAMPC\cite{sac19}, a distributed, adaptive-horizon and adaptive-neighborhood algorithm for solving the stochastic reachability problem in multi-agent systems; specifically the flocking problem modeled as an MDP. In DAMPC, at each time step, every agent first calls a centralized, adaptive-horizon model-predictive control (AMPC) algorithm to obtain an optimal solution for its local neighborhood. Second, the agents derive the flock-wide optimal solution through a sequence of consensus rounds. Third, the neighborhood is adaptively resized using a flock-wide cost-based Lyapunov function. In this way DAMPC improves efficiency without compromising convergence. The proof of statistical global convergence is non-trivial and involves showing that $V$ follows a monotonically decreasing trajectory despite potential fluctuations in cost and neighborhood size.

We evaluate DAMPC's performance using statistical model checking, showing that DAMPC achieves considerable speed-up over AMPC (two-fold in some cases) with only a slightly lower convergence rate.  Smaller average neighborhood size and lookahead horizon demonstrate the benefits of the DAMPC approach for stochastic reachability problems involving any controllable multi-agent system that possesses a cost function.

Inspired by the emerging problem of CPS security, we lastly introduce the concept of \emph{controller-attacker games}~\cite{tiwari17}: a two-player stochastic game involving a controller and an attacker, which have antagonistic objectives.  A controller-attacker game is formulated in terms of an MDP, with the controller and the attacker jointly determining the MDP's transition probabilities. We also introduce \emph{V-formation games}, a class of controller-attacker games where the goal of the controller is to maneuver the plant (a simple model of flocking dynamics) into a V-formation, and the goal of the attacker is to prevent the controller from doing so.  Controllers in V-formation games utilize
AMPC, giving them extraordinary power: we prove that under certain controllability conditions, an AMPC controller can attain V-formation with probability~1. 

We evaluate AMPC's performance on V-formation games using statistical model checking.  Our results show that (a)~as we increase the power of the attacker, the AMPC controller adapts by suitably increasing its horizon, and thus demonstrates resiliency to a variety of attacks; and (b)~an intelligent attacker can significantly outperform its naive counterpart.

The rest of the paper is organized as follows. Section~\ref{sec:bg} provides background content in the form of our dynamic model of V-formation, stochastic reachability, and PSO. Sections~\ref{sec:ares}-\ref{sec:cag} present the ARES algorithm, the DAMPC algorithm, and controller-attacker games for V-formation, respectively. Section~\ref{sec:concl} offers our concluding remarks.

This paper was written on the occasion of Jos Baeten's retirement as general director of CWI and professor of theory of computing of ILLC.  Jos was a highly influential collaborator of the third author (Smolka), and remains a good friend and colleague.  Jos's feedback to Smolka on the invited talk he gave on V-formation at CONQUEST 2016 was an important impetus for moving the work forward.

\section{Background}
\label{sec:bg}


This section introduces the basic concepts and techniques needed to formulate and derive our results.

\subsection{Dynamic Model for V-formation}
\label{subsec:dynmodel}



 In our flocking model, each bird in the flock is modeled using four variables: a 2-dimensional vector $\xv$ denoting the position of the bird in a 2D space, and a 2-dimensional vector $\vv$ denoting the velocity of the bird. We use $s=\{\xv_i, \vv_i\}_{i=1}^B$ to denote a state of a flock with $B$ birds. 
 The \emph{control actions} of each bird are 2-dimensional accelerations $\va$ and 2-dimensional position displacements $\vd$ (see discussion of $\va$ and $\vd$ below). Both are random variables.

Let $\xv_i(t),\vv_i(t),\va_i(t)$, and $\vd_i(t)$ respectively denote the position, velocity, acceleration, and displacement of the $i$-th bird at time $t$, $1\leqslant i \leqslant B$.
%
%
The behavior of bird $i$ in discrete time is modeled as follows:
\vspace*{-1mm}\begin{eqnarray}
\label{eq:trans}
 \xv_i(t + 1) &=& \xv_i(t) + \vv_i(t)\label{eq:x} + \vd_i(t) \qquad  \nonumber\\
 \vv_i(t + 1) &=& \vv_i(t) + \va_i(t)\label{eq:v} 
 \label{eq:withdist}
\end{eqnarray}
The next state of the flock is jointly determined by the accelerations and the displacements based on the current state following Eq.~\ref{eq:trans}.

Every bird in our model~\cite{grosu2014isola} moves in 2-dimensional space performing acceleration actions determined by a global controller. When there is no external disturbance, the displacement term is zero and the equations are:

\vspace*{-5mm}
\begin{align}
\xv_i(t + 1) &= \xv_i(t) + \vv_i(t)\notag\\
\vv_i(t + 1) &= \vv_i (t)+ \va_i(t)
\label{eq:nodist}
\end{align}

\vspace*{-1mm} The controller detects the positions and velocities 
of all birds through sensors, and uses this information to compute an optimal 
acceleration for the entire flock.  A bird uses its own component of the
solution to update its velocity and position.

We extend this discrete-time dynamical model to a (deterministic) MDP by adding a
cost (fitness) function\footnote{A classic 
MDP~\cite{russellnorvig} is obtained by adding sensor/actuator or wind-gust noise, which are the case we are addressing in the follow-up work.} based on the following metrics inspired by~\cite{yang2016love}:

\begin{itemize}
	\item \emph{Clear View} ($\CV$). A bird's visual field is a cone with 
    angle $\theta$ that can be blocked by the wings of other birds.  We define
    the clear-view metric by accumulating the percentage of a bird's visual 
    field that is blocked by other birds.  Fig.~\ref{fig:fitness} (left) illustrates 
    the calculation of the clear-view metric.  
        Let $B_{ij}(\xv_i, \vv_i, \xv_j)$ be the part of the angle subtended by the wing of Bird $j$ on the eye of Bird $i$ that intersects with Bird $i$'s visual cone with angle $\theta$. Then,
        the clear view for Bird $i$, 
        $\CV_i(\xv, \vv)$, is defined as $|\cup_{j\neq i} B_{ij}(\xv_i, \vv_i, \xv_j)|/\theta$, and the total 
        clear view, $\CV(\xv,\vv)$, is defined as $\sum_i \CV_i(\xv, \vv)$.
        The optimal value in a V-formation 
    is $\CV^*{=}\,0$, as all birds have a clear view.
    Note that the value $B_{ij}$ can be computed using Bird $i$'s velocity and position, and Bird $j$'s position using standard trigonometric functions.
\item \emph{Velocity Matching} ($\VM$). The accumulated 
	differences between the velocity of each bird and all other birds, 
    summed up over all birds in the flock defines $\VM$. Fig.~\ref{fig:fitness} (middle) depicts the values of $\VM$ in a velocity-unmatched flock. 
    Formally, $\VM(\xv, \vv) = \sum_{i > j} (||\vv_i-\vv_j||/(||\vv_i||+||\vv_j||))^2$.
    The optimal
    value in a V-formation is $\VM^*{=}\,0$, as all birds will have the same 
    velocity (thus maintaining the V-formation).
	\vspace*{1mm}\item \emph{Upwash Benefit} ($\UB$). The trailing upwash is 
    generated near the wingtips of a bird, while downwash is generated near 
    the center of a bird.  We accumulate all birds' upwash benefits using a 
    Gaussian-like model of the upwash and downwash region, as shown in 
    Fig.~\ref{fig:fitness} (right) for the right wing. 
        Let $h_{ij}$ be the projection of the 
        vector $\xv_j - \xv_i$ along the wing-span of Bird $i$.
        Similarly, let $g_{ij}$ be the projection of $\xv_j - \xv_i$ 
        along the direction of $\vv_i$.
        Specifically, the upwash benefit $\UB_{ij}$ for Bird $i$ coming from Bird $j$ is given by
        \begin{eqnarray*}
         \UB_{ij} = \left\{
             \begin{array}{ll}
                 \alpha  S(|h_{ij}|) G(h_{ij}, g_{ij}, \mu_1, \Sigma_1) & \mbox{ if } |h_{ij}| \geq \frac{(4-\pi)w}{8} \wedge g_{ij} > 0
             \\
                 S(|h_{ij}|) G(h_{ij}, g_{ij}, \mu_1, \Sigma_1) & \mbox{ if } |h_{ij}| < \frac{(4-\pi)w}{8} \wedge g_{ij} > 0
             \\
               0 & \mbox{ otherwise}
             \end{array}
             \right.
        \end{eqnarray*}
        where $S(z) = \mathtt{erf}(2\sqrt{2}(z - \frac{(4-\pi)w}{8}))$ is the error function, which is a smooth approximation of the sign function, 
        $G_1(\vec{z}, \Sigma) = e^{(-\frac{1}{2}(\vec{z}^T \Sigma^{-1}\vec{z}))}$ is a 2D-Gaussian with mean at the origin, 
        and
        $G(y,z,\mu,\Sigma) = G_1( [|y|,|z|] - \mu, \Sigma)$ is a 2D-Gaussian shifted so that the mean is $\mu$.
        The parameter $w$ is the wing span, and
        $\mu_1 = [(12+\pi)w/16,1]$ is the relative position where upwash benefit is maximized. 
        The total upwash benefit, $\UB_i$, for Bird $i$ is $\sum_{j\neq 1}\UB_{ij}$.
        The maximum upwash a bird can obtain is upper-bounded by~1.  
        Since we are working with cost (that we want to minimize), we define
        $\UB(\xv, \vv) = \sum_i (1 - \min(\UB_i, 1))$. 
    The optimal value for $\UB$ in a V-formation is $\UB^*{=}\,1$, as the leader does not receive any upwash.
\end{itemize}

Finding smooth and continuous formulations 
of the fitness metrics is a key element of solving optimization problems.  The PSO 
algorithm has a very low probability of finding an optimal solution if 
the fitness metric is not well-designed.

\begin{figure}[t]
\centering	
\includegraphics[width=.28\textwidth]{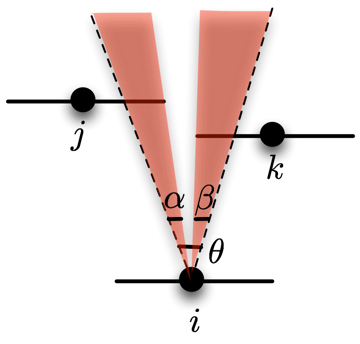}
\includegraphics[width=.3\textwidth]{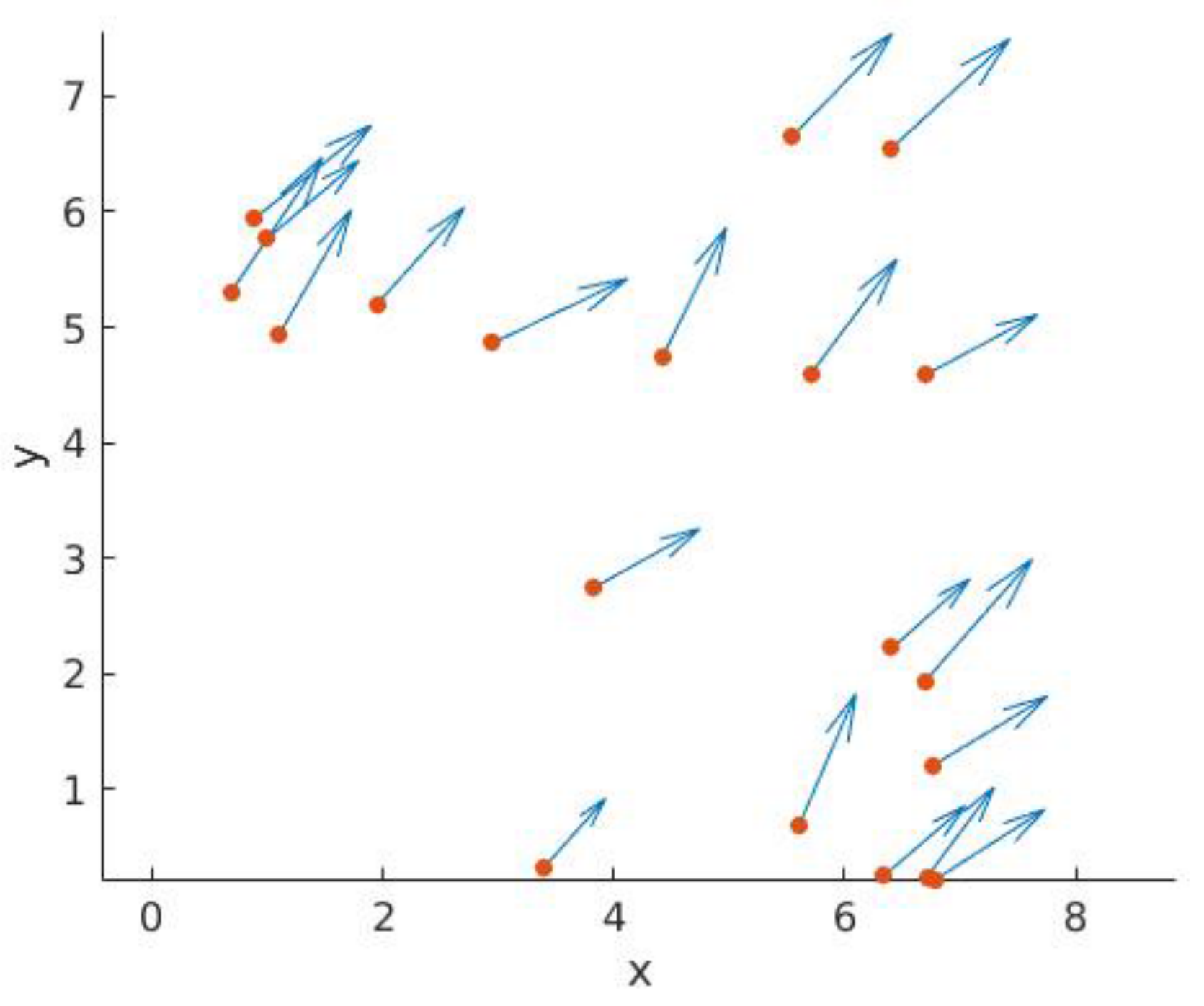}
\includegraphics[width=.36\textwidth]{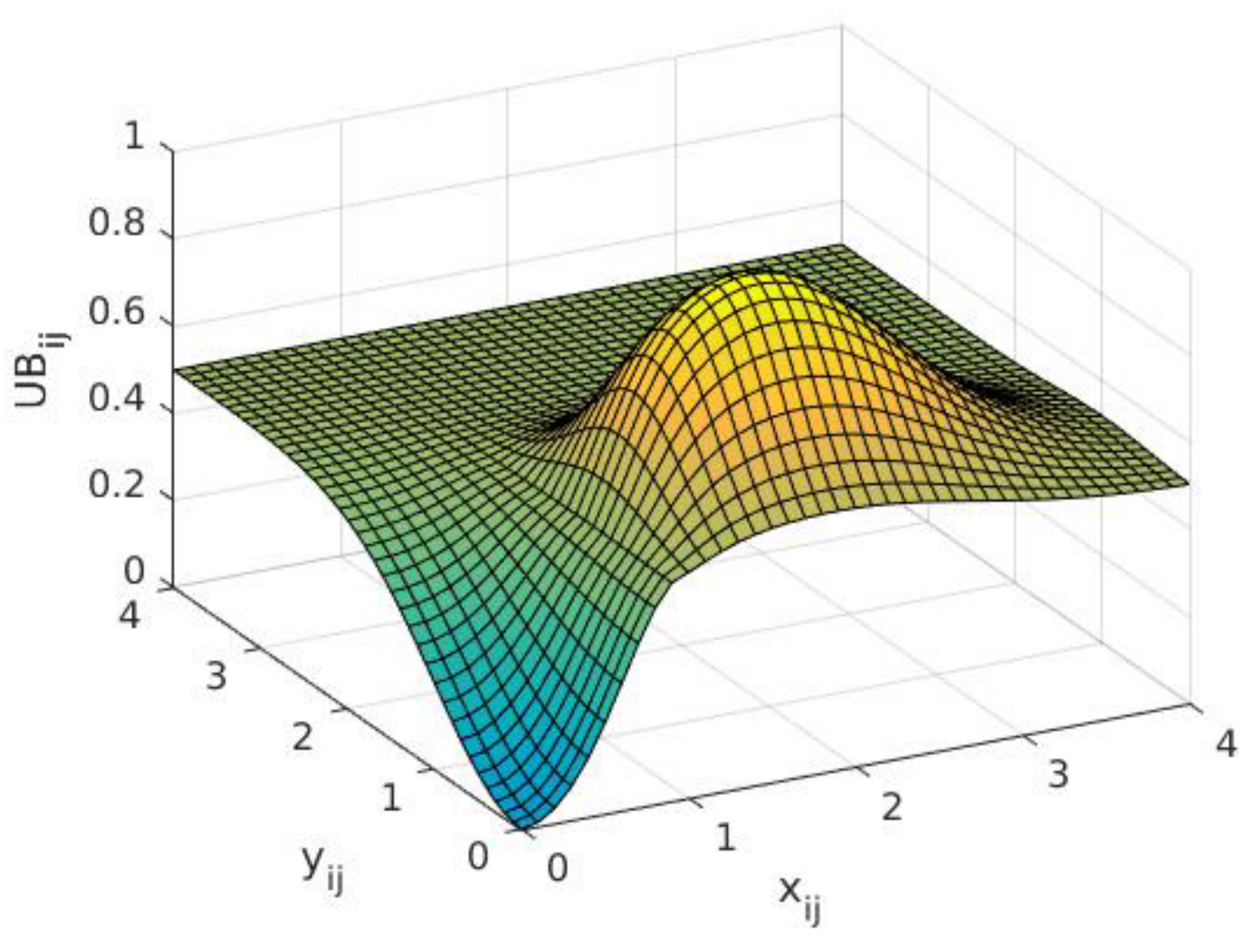}
\vspace*{-1mm}
\caption{
Illustration of the clear view ($\CV$), velocity matching  
($\VM$), and upwash benefit ($\UB$) metrics. 
Left: Bird $i$'s view is partially blocked by birds $j$ 
and $k$. Hence, its clear view is $\CV\,{=}\,(\alpha\,{+}\,\beta)/\theta$. 
Middle: A flock and its unaligned bird velocities results in a 
velocity-matching metric $\VM\,{=}\,6.2805$.  In contrast, $\VM\,{=}\,0$
when the velocities of all birds are aligned. 
Right: Illustration of the (right-wing) upwash benefit bird $i$ receives from bird $j$ depending on how it is positioned behind bird $j$.  Note that  
bird $j$'s downwash region is directly behind it.
}
\label{fig:fitness}
\vspace*{-4mm}
\end{figure}
Let $\boldsymbol{c}(t)\,{=}\,\{\boldsymbol{c}_i(t)\}_{i=1}^b\,{=}\,\{\xv_i(t), \vv_i(t)\}_{i=1}^b\,{\in}\,\mathbb{R}$ be a flock 
configuration at time-step $t$. Given the above metrics, the overall fitness (cost) metric $J$ 
is of a sum-of-squares combination of $\VM$, $\CV$, and $\UB$ defined as follows:
\vspace*{-1mm}
\begin{align}
J(\boldsymbol{c}(t),\va^h(t),{h}) = (\CV(\boldsymbol{c}_{\va}^{h}(t))-\CV^*)^2 &+ 
(\VM(\boldsymbol{c}_{\va}^{h}(t))-\VM^*)^2 \nonumber \\ & +(\UB(\boldsymbol{c}_{\va}^{h}(t))-\UB^*)^2,
\label{eq:fitness}
\end{align}
{}where ${h}$ is the receding prediction horizon (RPH), $\va^h(t)\,{\in}\,\mathbb{R}$ is a sequence of accelerations of length ${h}$, and $\boldsymbol{c}_{\va}^{h}(t)$ is 
the configuration reached after applying $\va^h(t)$ to $\boldsymbol{c}(t)$.
Formally, we have\vspace*{-2mm}
\begin{align}
\boldsymbol{c}_{\va}^{h}(t)= \{\xv_{\va}^{h}(t), \vv_{\va}^{h}(t)\} = \{\xv(t)+\sum_{\tau=1}^{{h}(t)}\vv(t+\tau), \vv(t)+\sum_{\tau=1}^{{h}(t)} \va^\tau(t) \},
\end{align}
where $\va^\tau(t)$ is the $\tau$th acceleration of $\va^h(t)$.
As discussed further in Section~\ref{sec:ares}, we allow  RPH ${h}(t)$ to be \emph{adaptive} in nature. 

The fitness function $J$ has an optimal value of $0$ in a perfect V-formation. Thus, there is a need to perform flock-wide minimization of $J$ at each time-step $t$ to obtain an optimal plan of  length $h$ of acceleration actions:\vspace*{-3mm} 

\begin{align}
&\textbf{opt-$\va$}^{h}(t)=\{\textbf{opt-$\va$}_i^{h}(t)\}_{i=1}^{b}=\argmin_{\va^h(t)}J(\boldsymbol{c}(t),\va^h(t),{h}).
\label{eq:opt}
\end{align}
\vspace*{-3mm}{} 

 The optimization is subject to the following constraints on the maximum 
velocities and accelerations: $||\vv_i(t)||\,{\leqslant}\,\vv_{max},
||\va^h_i(t)||\,{\leqslant}\,\rho||\vv_i(t)||$ $\forall$ $i\,{\in}\,\{1,\ldots,b\}$,
where $\vv_{max}$ is a constant and $\rho\,{\in}\,(0,1)$. The above constraints prevent us from using mixed-integer programming, we might, however, compare our solution to other continuous optimization techniques in the future.
The initial positions and velocities of each bird are selected at random 
within certain ranges, and limited such that the distance between any 
two birds is greater than a (collision) constant $d_{min}$, and small
enough for all birds, except for at most one, to feel the $\UB$. 
 
\subsection{V-Formation MDP}
\label{subsec:mdp}

This section defines Markov Decision Processes (MDPs) and the corresponding MDP formulated by Lukina et al.~\cite{lukina-tacas17} for the V-formation problem.

\begin{definition}
A \textbf{Markov decision process} (MDP) is a 5-tuple $\M=(S,A,T,J,I)$ consisting of a set of states $S$, a set of actions $A$, a transition function $T: S\,{\times}\,A\,{\times}\,S\,{\mapsto}\,[0,1]$, where $T(\vs,a,\vs^\prime)$ is the probability of transitioning from state $\vs$ to state $\vs'$ under action $\va$, a cost function $J:S\,{\mapsto}\,\mathbb{R}$, where $J(\vs)$ is the cost associated with state $\vs$, and an initial state distribution $I: S\,{\mapsto}\,[0,1]$.
\end{definition}

 The \emph{MDP $\M$ modeling a flock} of $B$ birds is defined as follows. The set of states $S$ is $S = \mathbb{R}^{4B}$, as each bird has a $2$D position and a $2$D velocity vector, and the flock contains $B$ birds. The set of actions $A$ is $A = \mathbb{R}^{2B}$, as each bird takes a $2$D acceleration action and there are $B$ birds. The cost function  $J$ is defined by Eq.~\ref{eq:fitness}. The transition function $T$ is defined by Eq.~\ref{eq:v}. As the acceleration vector $\va_i(t)$ for bird $i$ at time $t$ is a random variable, the state vector $\vs_i=\{\xv_i(t+1)$, $\vv_i(t+1)\}$ is also a random variable. The initial state distribution $I$ is a uniform distribution from a region of state space where all birds have positions and velocities in a range defined by fixed lower and upper bounds.

\subsection{Stochastic Reachability Problem}
\label{sec:reach}

Given the stochasticity introduced by PSO, the V-formation problem can be formulated in terms of a reachability problem for the Markov chain induced by the composition of a Markov decision process (MDP) and a controller.  

Before we can define traces, or executions, of $\M$, we need to fix a controller, or strategy, that determines which action from $A$ to use at any given state of the system. We focus on randomized strategies. A \emph{randomized strategy (controller)} $\sigma$ over $\M$ is a function of the form $\sigma: S\,{\mapsto}\,\PD(A)$, where $\PD(A)$ is the set of probability distributions over $A$. That is, $\sigma$ takes a state $\vs$ and returns an action consistent with the probability distribution $\sigma(\vs)$. Applying a policy $\sigma$ to the MDP $\M$ defines the Markov chain. $\M_{\sigma}$. We use the terms strategy and controller interchangeably.

In the bird-flocking problem, a controller would be a function that determines the accelerations for all the birds given their current positions and velocities. Once we fix a controller, we can iteratively use it to (probabilistically) select a sequence of flock accelerations. The goal is to generate a sequence of actions that takes an MDP from an initial state $\vs$ to a state $\vs^{*}$ with $J(\vs^{*})\,{\leqslant}\,\varphi$. 


\begin{definition} Let $\M\,{=}\,(S,A,T,J,I)$ be an MDP, and let $G \subseteq S$ be the set of goal states $G\,{=}\,\{\vs | J(\vs)\,{\leqslant}\,\varphi\}$ of $\M$. The \textbf{stochastic reachability problem} is to design a controller $\sigma: S\,{\mapsto}\,\PD(A)$ for $\M$ such that for a given $\delta$, the probability of the underlying Markov chain $\M_\sigma$ to reach a state in $G$ in $m$ steps (for a given $m$) starting from an initial state, is at least $1-\delta$.
\end{definition}

We approach the stochastic reachability problem by designing a controller and quantifying its probability of success in reaching the goal states.

\subsection{Particle Swarm Optimization}
Particle Swarm Optimization (PSO) is a randomized approximation algorithm for computing the value of a parameter minimizing a possibly nonlinear cost (fitness) function.  Interestingly, PSO itself is inspired by bird flocking~\cite{Kennedy95particleswarm}.  Hence, PSO assumes that it works with a flock of birds. 

Note, however, that in our running example, these birds are ``acceleration birds'' (or particles), and not the actual birds in the flock. Each bird has the same goal, finding food (reward), but none of them knows the location of the food. However, every bird knows the distance (horizon) to the food location. PSO works
by moving each bird preferentially toward the bird closest to food.

The work delineated in this paper uses Matlab-Toolbox $\texttt{particleswarm}$, which performs the classical version of PSO. This PSO creates a swarm of particles, of size say ${p}$, uniformly at random within a given bound on their positions and velocities. Note that in our example, each particle represents itself a flock of bird-acceleration sequences $\{\va_i^{{h}}\}_{i=1}^b$, where ${h}$ is the current length of the receding horizon. PSO further chooses a neighborhood of a random size for each particle $j$, $j\,{=}\,\{1,\ldots,p\}$, and computes the fitness of each particle. Based on the fitness values, PSO stores two vectors for $j$: its so-far personal-best position $\mathbf{x}_{P}^j(t)$, and its fittest neighbor's position $\mathbf{x}_{G}^j(t)$. The positions and velocities of each particle $j$ in the particle swarm $1\,{\leqslant}\,j\,{\leqslant}\,p$ are updated according to the following rule:

\vspace*{-4mm}
\begin{align}
\mathbf{v}^j(t+1) = \omega\cdot\mathbf{v}^j(t) &+ y_1\cdot \mathbf{u_1}(t+1)\otimes(\mathbf{x}_{P}^j(t)-\mathbf{x}^j(t))  \nonumber \\
&+ y_2\cdot \mathbf{u_2}(t+1)\otimes(\mathbf{x}_{G}^j(t)-\mathbf{x}^j(t))
\label{eq:swarm}
\end{align}
where $\omega$ is \emph{inertia weight}, which determines the trade-off between global and local exploration of the swarm (the value of $\omega$ is proportional to the exploration range); $y_1$ and $y_2$ are \emph{self adjustment} and \emph{social adjustment}, respectively; $\mathbf{u_1},\mathbf{u_2}\,{\in}\,{\rm Uniform}(0,1)$ are randomization factors; and $\otimes$ is the vector dot product, that is, $\forall$ random vector $\mathbf{z}$: $(\mathbf{z}_1,\ldots,\mathbf{z}_b)\otimes(\mathbf{x}_1^j,\ldots,\mathbf{x}_b^j)=(\mathbf{z}_1\mathbf{x}_1^j,\ldots,\mathbf{z}_b\mathbf{x}_b^j)$. 

If the fitness value 
for $\mathbf{x}^j(t+1)\,{=}\,\mathbf{x}^j(t)\,{+}\,\mathbf{v}^j(t+1)$ is lower than the one for $\mathbf{x}_{P}^j(t)$, then $\mathbf{x}^j(t+1)$ is assigned to $\mathbf{x}_{P}^j(t+1)$. The particle with the best fitness over the whole swarm becomes a global best for the next iteration. The procedure is repeated until the number of iterations reaches its maximum, the time elapses, or the minimum criteria is satisfied. For our bird-flock example we obtain in this way the best acceleration.



\section{Adaptive Receding-Horizon Synthesis of Optimal Plans (ARES)}
\label{sec:ares}

ARES~\cite{lukina-tacas17} is a general \emph{adaptive, receding-horizon synthesis algorithm} (ARES) that, given an MDP and one of its initial states, generates an optimal plan (action sequence) taking that state to a state whose cost is below a desired threshold. ARES implicitly defines an \emph{optimal, online policy-synthesis algorithm}, assuming plan generation can be performed in real-time.  ARES can alternatively be viewed as a model-predictive control (MPC) algorithm that utilizes an adaptive receding horizon, a technique we refer to as Adaptive MPC (AMPC).

ARES makes repeated use of PSO~\cite{Kennedy95particleswarm} to effectively generate a plan. This was in principle unnecessary, as one could generate an optimal plan by calling PSO only once, with a maximum plan-length horizon.  Such an approach, however, is in most cases impractical, as every unfolding of the MDP adds a number of new dimensions to the search space.  Consequently, to obtain adequate coverage of this space, one needs a very large number of particles, a number that is either going to exhaust available memory or require a prohibitive amount of time to find an optimal plan.


\subsection{The ARES Algorithm}
One could in principle solve the optimization problem defined in Sections~\ref{subsec:dynmodel} and~\ref{subsec:mdp} by calling PSO only once, with a horizon $h$ in $\mathcal{M}$ equaling the maximum length $m$ allowed for a plan. This approach, however, tends to lead to very large search spaces, and is in most cases intractable. Indeed, preliminary experiments with this technique applied to our running example could not generate any convergent plan.  
\begin{figure}[t]
\centering
\includegraphics[height=2in]{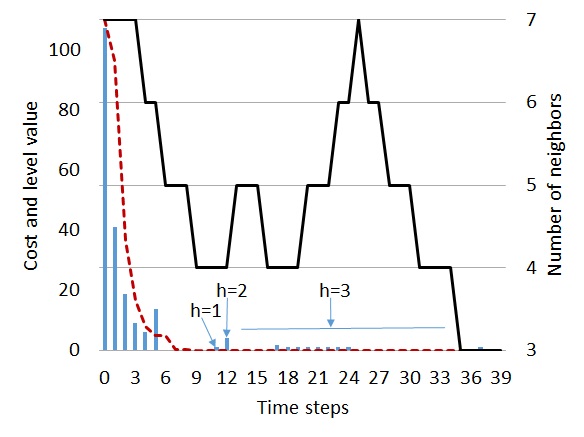}
\caption{Blue bars are the values of the cost function in every time step. Red dashed line is the cost-based Lyapunov function used for horizon and neighborhood adaptation. Black solid line is neighborhood resizing for the next step given the current cost. 
}
\label{fig:example_local}
\end{figure}

A more tractable approach is to make repeated calls to PSO with a small horizon length $h$.  The question is how small $h$ can be. \emph{The current practice in model-predictive control (MPC) is to use a fixed} $h$, $1\,{\leqslant}\,h\,{\leqslant}\,3$ (see the outer loop of Fig.~\ref{fig:approach}, where resampling and conditional branches are disregarded). Unfortunately, this forces the selection of \emph{locally-optimal plans} (of size less than three) in each call, and there is \emph{no guarantee of convergence} when joining them together.  In fact, in our running example, we were able to find plans leading to a V-formation in only $45\%$ of the time for $10,000$ random initial flocks.

Inspired by Importance Splitting (see Fig.~\ref{fig:levels}~(right) and Fig.~\ref{fig:approach}), we introduce the notion of a \emph{level-based horizon}, where level $\ell_0$ equals the cost of the initial state, and level $\ell_m$ equals the threshold $\varphi$.  Intuitively, by using an asymptotic cost-convergence function ranging from $\ell_0$ to $\ell_{m}$, and dividing its graph in $m$ equal segments, we can determine on the vertical axis a sequence of levels ensuring convergence. 

The asymptotic function ARES implements is essentially $\ell_{i}\,{=}\,\ell_0\,(m-i){/}\,m$, but specifically tuned for each particle. Formally, if particle $k$ has previously reached level equaling $J_k(s_{i-1})$, then its next target level is within the distance $\Delta_k\,{=}\,J_k(s_{i-1}){/}(m\,{-}\,i\,{+}\,1)$. In Fig.~\ref{fig:approach}, after passing the thresholds assigned to them, values of the cost function in the current state $s_i$ are sorted in ascending order $\{\widehat{J}_{k}\}_{k=1}^n$. The lowest cost $\widehat{J}_1$ should be apart from the previous level $\ell_{i-1}$ at least on its $\Delta_1$  for the algorithm to proceed to the next level $\ell_i\,{:=}\,\widehat{J}_1$.
\begin{figure}[t]
	\centering
	\includegraphics[height=4.15in]{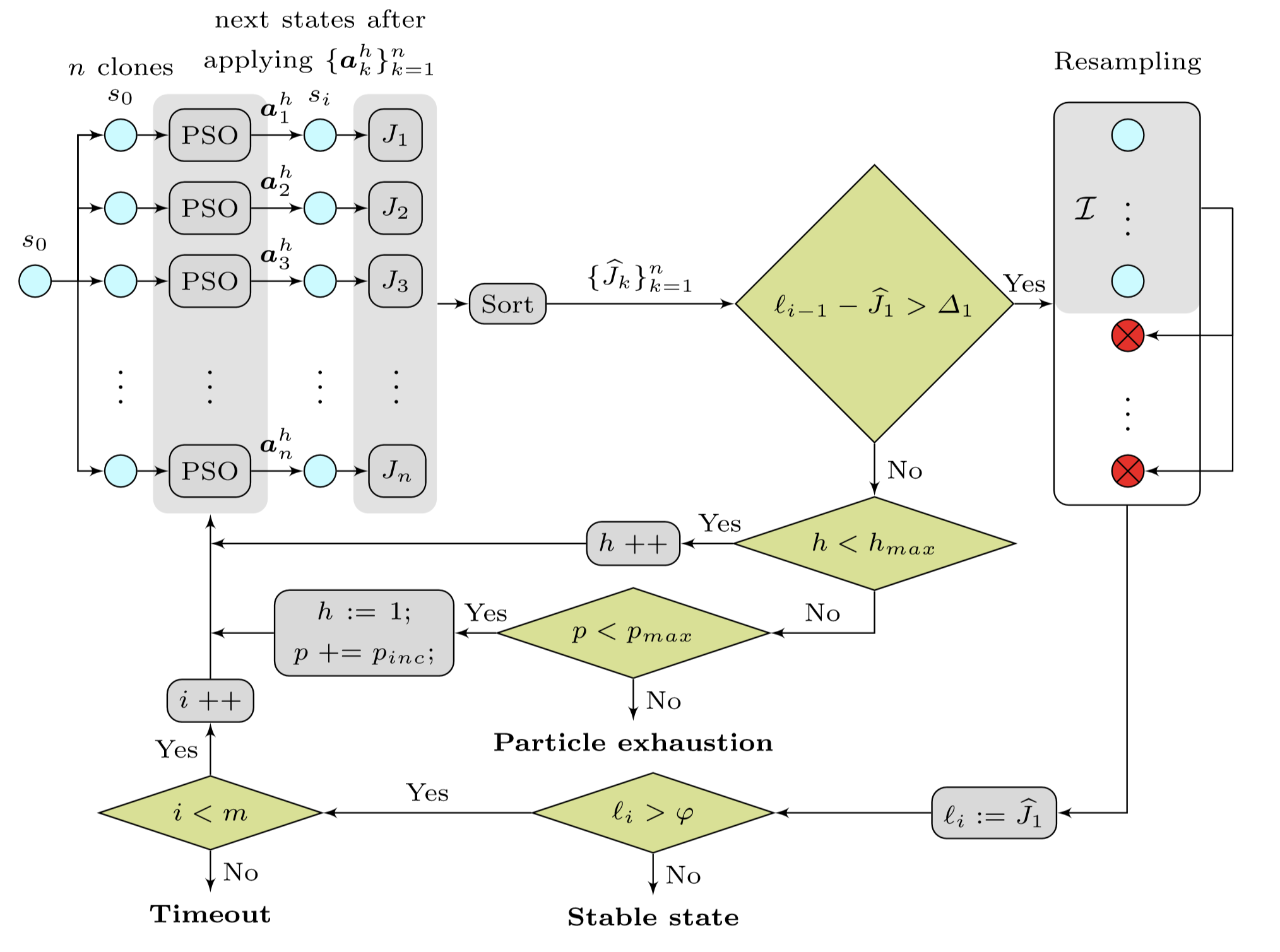}
    \vspace*{-2mm}
\caption{Graphical representation of ARES.}
	\label{fig:approach}
    \vspace*{3mm}
\end{figure}
\begin{algorithm}[b]
	\SetKwFunction{Cost}{Cost}
	\SetKwFunction{ImportanceSplitting}{ImportanceSplitting}
	\SetKwFunction{particleswarm}{particleswarm}
    \SetKwFunction{Simulate}{Simulate}
	\SetKwInOut{Input}{Input}
	\SetKwInOut{Output}{Output}
\ForEach {$\mathcal{M}_k\in \mathcal{M}$}
{
    $[\va^{h}_k,\mathcal{M}^{h}_k] \leftarrow$ \particleswarm{$\mathcal{M}_k,{p},h$}; \textit{// use 
    PSO in 
    order to determine best next action for the MDP $\mathcal{M}_k$ with RPH $h$\\}
    ${J}_k(s_i)\leftarrow$ 
    \Cost{$\mathcal{M}_k^h,\va^{h}_k,{h}$};
    \textit{// calculate cost function if applying the sequence of optimal actions of length $h$}\\
    \If{${J}_k(s_{i-1})-{J}_k(s_i)>\Delta_k$}
    {
    $\Delta_k\leftarrow{J}_k(s_{i})/(m-i);$ \textit{// new level-threshold}\\
    }
}
	\caption{Simulate ($\mathcal{M},h,i,\{\Delta_k,{J}_k(s_{i-1})\}_{k=1}^n$)}
	\label{alg:sim}
\end{algorithm}
\setlength{\floatsep}{0.1cm}

 The levels serve two purposes.  First, they implicitly define a Lyapunov function, which guarantees convergence.  If desired, this function can be explicitly generated for all states, up to some topological equivalence.  Second, the levels $\ell_{i}$ help PSO overcome local minima    (see Fig.~\ref{fig:levels}~(left)).  If reaching a next level requires PSO to temporarily pass over a state-cost ridge, then ARES incrementally increases the size of the horizon $h$, up to a maximum size $h_{max}$. For particle $k$, passing the thresholds $\Delta_k$ means that it reaches a new level, and the definition of  $\Delta_k$ ensures a smooth degradation of its threshold.

\begin{figure}[t]
	\centering
    \includegraphics[height=2.75in]{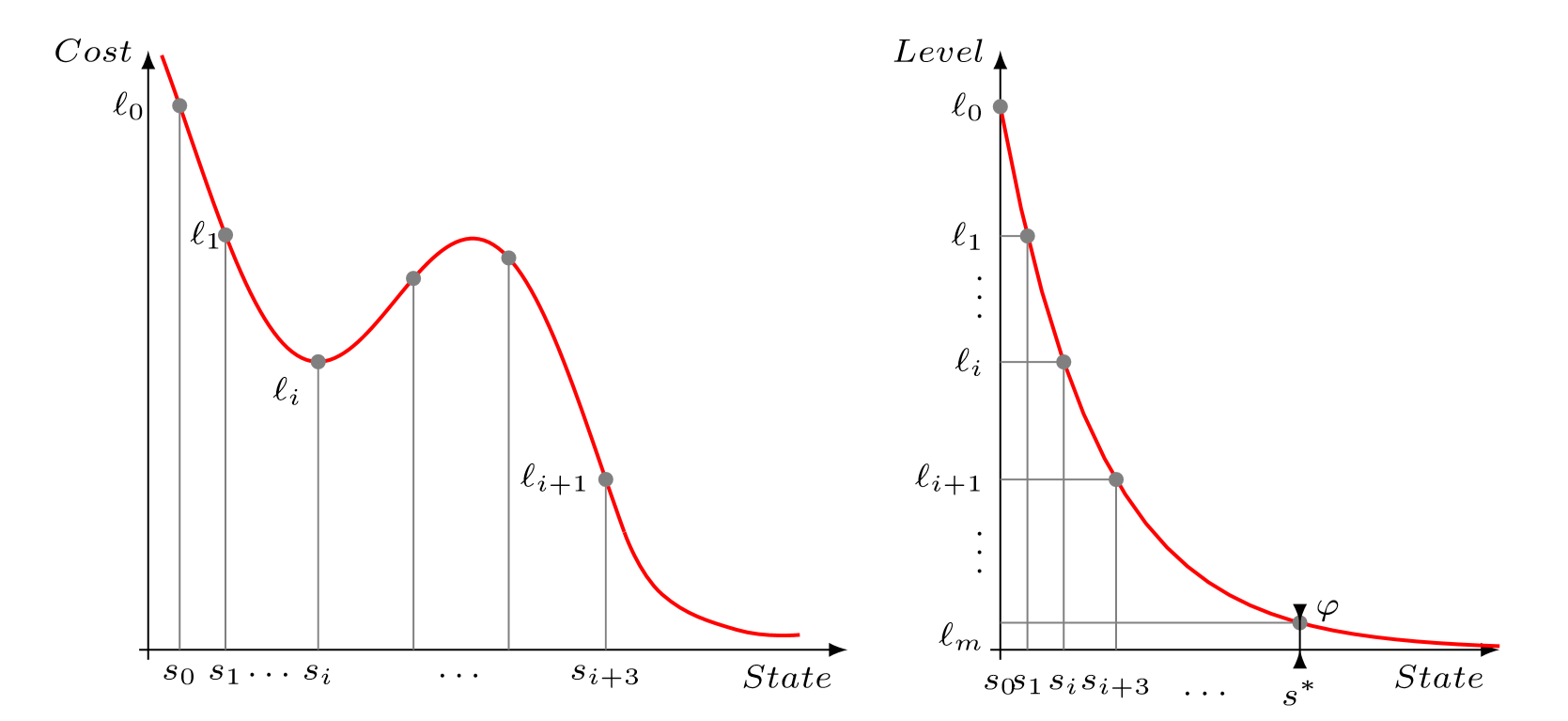}\vspace{-3ex}
\caption{Left: If state $s_0$ has cost $\ell_0$, and its successor-state $s_{1}$ has cost less than $\ell_{1}$, then a horizon of length~1 is appropriate. However, if $s_i$ has a local-minimum cost $\ell_i$, one has to pass over the cost ridge in order to reach level $\ell_{i+1}$, and therefore ARES has to adaptively increase the horizon to~3. Right: The cost of the initial state defines $\ell_0$ and the given threshold $\varphi$ defines $\ell_{m}$. By choosing $m$ equal segments on an asymptotically converging (Lyapunov) function (where the number $m$ is empirically determined), one obtains on the vertical cost-axis the levels required for ARES to converge.}   
	\label{fig:levels}
\end{figure}

 Another idea imported from IS and shown in Fig.~\ref{fig:approach}, is to maintain $n$ clones $\{\mathcal{M}_k\}_{k=1}^n$ of the MDP $\mathcal{M}$ (and its initial state) at any time $t$, and run PSO, for a horizon $h$, on each $h$-unfolding $\mathcal{M}^h_k$ of them. This results in an action sequence $\va^{h}_k$ of length $h$ (see Algo.~\ref{alg:sim}). This approach allows us to call  PSO for each clone and desired horizon, with a very small number of particles $p$ per clone.  

\SetAlgoSkip{}
\begin{algorithm}[t]
\BlankLine
	$\mathcal{I}\leftarrow$ Sort ascending $\mathcal{M}^h_{k}$ by their current costs;  \textit{// find indexes of MDPs whose costs are below the median among all the clones}\\
	\For{$k=1$ \KwTo $n$}
	{
		\eIf{$k\notin\mathcal{I}$}
		{
			Sample $r$ uniformly at random from $\mathcal{I}$;
			$\mathcal{M}_k \leftarrow \mathcal{M}_r^h$;\\
		}
		{
			$\mathcal{M}_k\leftarrow\mathcal{M}^h_k$; \textit{// Keep more successful MDPs unchanged} 
		}
	}
	\caption{Resample ($\{\mathcal{M}_k^h,{J}_k(s_{i})\}_{k=1}^n$)}
	\label{alg:resample}
\SetAlgoSkip{}
\end{algorithm}
\setlength{\floatsep}{0.1cm}

To check which particles have overcome their associated thresholds, we sort the particles according to their current cost, and split them in two sets: the successful set, having the indexes $\mathcal{I}$ and whose costs are lower than the median among all clones; and the unsuccessful set with indexes in $\{1,{\ldots},n\}\,{\setminus}\mathcal{I}$, which are discarded. The unsuccessful ones are further replenished, by sampling uniformly at random from the successful set $\mathcal{I}$ (see Algo.~\ref{alg:resample}).

The number of particles is increased $p\,{=}\,p\,{+}\,p_{inc}$ if no clone reaches a next level, for all horizons chosen.  Once this happens, we reset the horizon to one, and repeat the process. In this way, we adaptively focus our resources on escaping from local minima.  From the last level, we choose the state $s^{*}$ with the minimal cost, and traverse all of its predecessor states to find an optimal plan comprised of actions $\{\va^i\}_{1\leqslant i\leqslant m}$ that led MDP $\mathcal{M}$ to the optimal state $s^*$. In our running example, we select a flock in V-formation, and traverse all its predecessor flocks. The overall procedure of ARES is shown in Algo.~\ref{alg:ares}.

\SetAlgoSkip{}
\begin{algorithm}[!ht]
	\SetKwFunction{Fitness}{Fitness}
	\SetKwFunction{Resample}{Resample}
	\SetKwFunction{particleswarm}{particleswarm}
    \SetKwFunction{Simulate}{Simulate}
	\SetKwInOut{Input}{Input}
	\SetKwInOut{Output}{Output}
	
	\Input{$\mathcal{M},\varphi,{p}_{start},{p}_{inc},{p}_{max},{h}_{max},m,n$}
	\Output{$\{\va^i\}_{1\leqslant i\leqslant\,m}$ \textit{// synthesized optimal plans}}
	\BlankLine
	Initialize $\ell_0\leftarrow\inf$; $\{J_k(s_0)\}_{k=1}^n\leftarrow\inf$; ${p}\leftarrow{p}_{start}$; $i\leftarrow 1$;  ${h}\leftarrow 1$; $\Delta_k\leftarrow 0$;
	\BlankLine
		\While{($\ell_i > \varphi)$ $\vee$ $(i
	< m)$}
	{
     	\textit{// find and apply best actions with RPH $h$}\\
		$[\{\va_k^h,J_k(s_i),\mathcal{M}^h_{k}\}_{k=1}^n]\leftarrow$\Simulate{$\mathcal{M},h,i,\{\Delta_k,{J}_k(s_{i-1})\}_{k=1}^n$};
		$\widehat{J}_1 \leftarrow 
		sort({J}_1(s_i),\ldots,{J}_n(s_i))$; \textit{// find minimum cost among all the clones}\\	
		\eIf{$\ell_{i-1}-\widehat{J}_1>\Delta_1$}
		{
			$\ell_i\leftarrow\widehat{J}_1$; \textit{// new level has been reached}\\
            $i \leftarrow i + 1$;
			${h} \leftarrow 1$;
			${p} \leftarrow {p}_{start}$; \textit{// reset adaptive parameters}\\
			$\{\mathcal{M}_k\}_{k=1}^n\leftarrow$ 
			\Resample{$\{\mathcal{M}_k^h,{J}_k(s_{i})\}_{k=1}^n$};
		}
		{
			\eIf{${h} < {h}_{max}$}
			{
				${h} \leftarrow {h} + 1$; \textit{// improve time exploration}\\
			}
			{
				\eIf{${p} < {p}_{max}$}
				{
					${h} \leftarrow 1$;
					${p} \leftarrow {p} + 
					{p}_{inc}$; \textit{// improve space exploration}\\
				}
				{break;}
			}	
		}
	}
    Take a clone in the state with minimum cost $\ell_i=J(s^*_i)\leqslant\varphi$ at the last level $i$;\\
    \ForEach{$i$}
    {
    $\{s_{i-1}^*,\va^i\}\leftarrow Pre(s_i^*);$ \textit{// find predecessor and corresponding action}\\
    }
	\caption{ARES}
	\label{alg:ares}
\end{algorithm}
\setlength{\floatsep}{0.1cm}

\begin{proposition}[Optimality and Minimality] (1) Let $\mathcal{M}$ be an MDP.
For any initial state $s_0$ of $\mathcal{M}$, ARES is able to solve the optimal-plan synthesis problem for $\mathcal{M}$ and $s_0$. 
(2)~An optimal choice of $m$ in function $\Delta_k$, for some particle $k$, ensures that ARES also generates the shortest optimal plan.
\end{proposition}
\vspace*{-2mm}
\begin{proof}[Sketch]
(1)~The dynamic-threshold function $\Delta_k$ ensures that the initial cost in $s_0$ is continuously decreased until it falls below $\varphi$.  Moreover, for an appropriate number of clones, by adaptively determining the horizon and the number of particles needed to overcome $\Delta_k$, ARES always converges, with probability 1, to an optimal state, given enough time and memory. (2)~This follows from  convergence property (1), and from the fact that ARES always gives preference to the shortest horizon while trying to overcome $\Delta_k$.
\end{proof}
%

The optimality referred to in the title of the paper is in the sense of (1).  One, however, can do even better than (1), in the sense of (2), by empirically determining parameter $m$ in the dynamic-threshold function $\Delta_k$.  Also note that ARES is an \emph{approximation algorithm}, and may therefore return non-minimal plans.  Even in these circumstances, however, the plans will still lead to an optimal state.  This is a V-formation in our flocking example.

\subsection{Evaluation of ARES}
\newcommand{\majExp}{{8,000}}
To assess the performance of our approach, we developed a simple simulation environment in Matlab. All experiments were run on an Intel Core i7-5820K CPU with 3.30 GHz and with 32GB RAM available. 

We performed numerous experiments with a varying number of birds. Unless stated otherwise, results refer to {\majExp} experiments with 7 birds with the following parameters: ${p}_{start}\,{=}\,10$, ${p}_{inc}\,{=}\,5$, ${p}_{max}\,{=}\,40$, $\ell_{max}\,{=}\,20$, ${h}_{max}\,{=}\,5$, $\varphi\,{=}\,10^{-3}$, and $n\,{=}\,20$. The initial configurations were generated independently uniformly at random subject to the following constraints: 
\begin{enumerate}\itemsep=0em
\item Position constraints: $\forall\:i\,{\in}\,\{1,{\ldots},7\}.\:\xv_i(0)\in[0,3]\times[0,3]$.
\vspace*{1mm}\item Velocity constraints: $\forall\:i\,{\in}\,\{1,{\ldots},7\}.\:\vv_i(0)\in[0.25,0.75]\times[0.25,0.75]$.
\end{enumerate}

\begin{figure}[t]
	\centering
	\includegraphics[width=.49\textwidth]{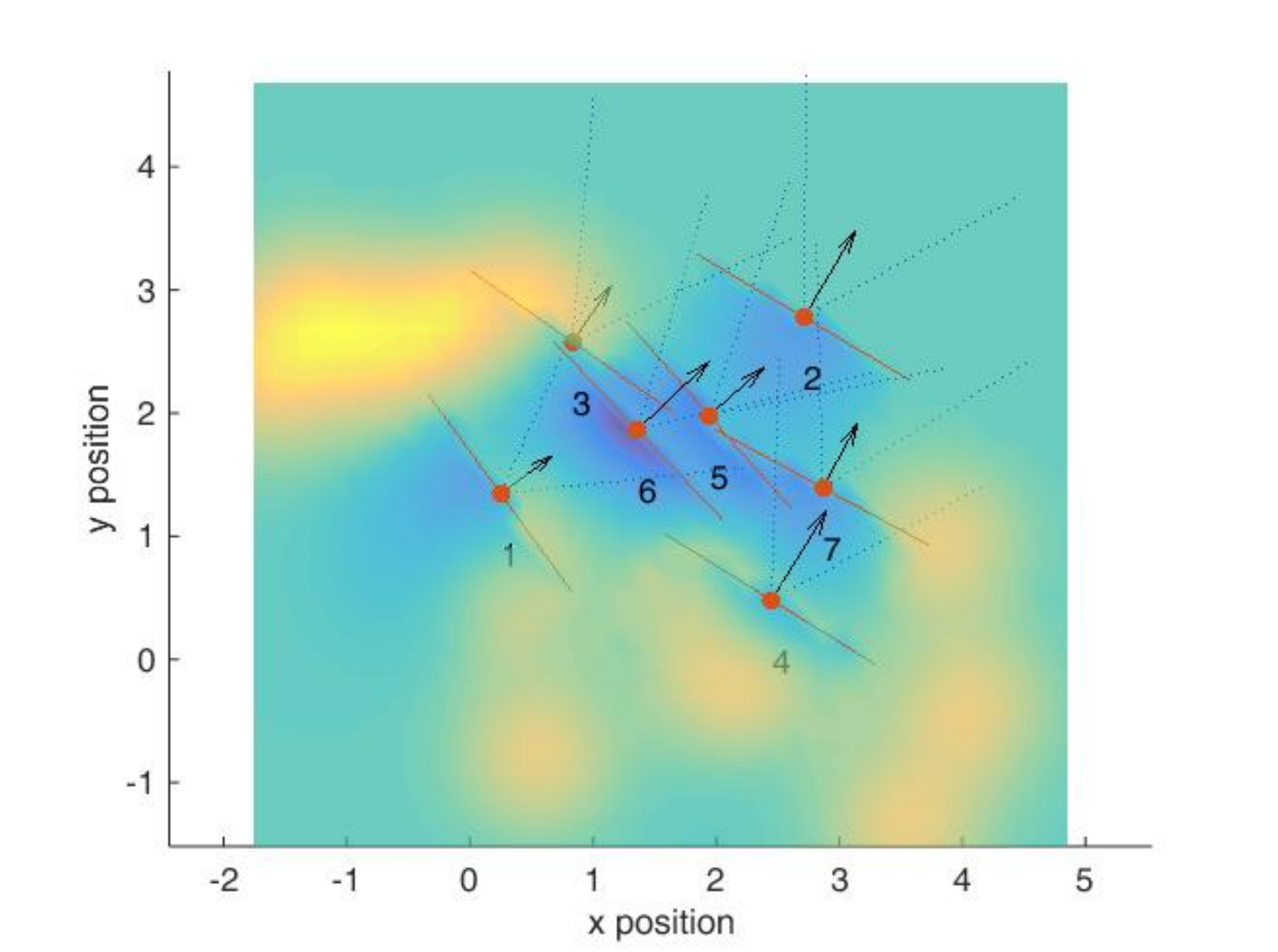}
	\includegraphics[width=.49\textwidth]{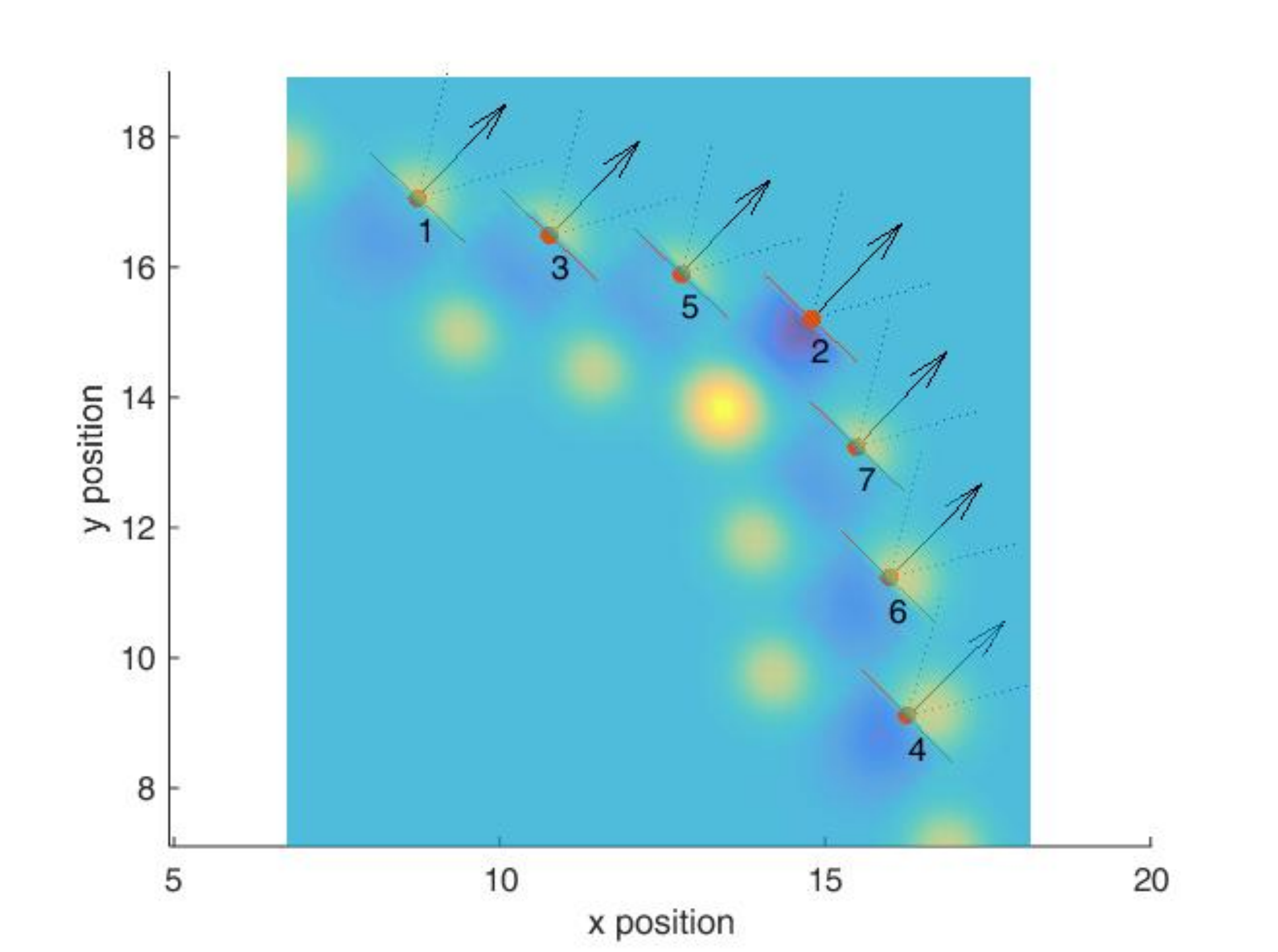}
    \vspace*{-2mm}
\caption{Left: Example of an arbitrary initial configuration of 7 birds. Right: The V-formation obtained by applying the plan generated by ARES. In the figures, we show the wings of the birds, bird orientations,  bird speeds (as scaled arrows), upwash regions in yellow, and downwash regions in dark blue.}
    \label{fig:form}
    \vspace*{3mm}
\end{figure}

\begin{table}
	\scriptsize
	\centering
	\caption{Overview of the results for \majExp$\:$experiments with 7 birds}
	\begin{tabular}{l c c c c c c c c }  
		\toprule
		& \multicolumn{4}{c}{{\textsc{Successful}}} & 
		\multicolumn{4}{c}{{\textsc{Total}}} \\
		\cmidrule(l){2-5}\cmidrule(l){6-9}
		No. Experiments~~~~ & \multicolumn{4}{c}{7573}  & 
		\multicolumn{4}{c}{8000} 
		\\
		\cmidrule(l){2-5}\cmidrule(l){6-9}
		& {\centering\textsc{Min}} & {\textsc{Max}}& {\centering\textsc{Avg}} & 
		{\textsc{Std}} 
		& {\textsc{Min}} & {\textsc{Max}} & {\textsc{Avg}} & 
		{\textsc{Std}} \\
		
		\midrule 
		Cost, ${J}$
			& 2.88$\cdot10^{-7}$	& 9$\cdot10^{-4}$ 	&	4$\cdot10^{-4}$  & 3$\cdot10^{-4}$	
			& 2.88$\cdot10^{-7}$	& 1.4840 	&	0.0282  & 0.1607   	\\
		Time, $t$
			& 23.14s 	& 310.83s 	& 63.55s	& 22.81s	
			& 23.14s  	& 661.46s 	& 64.85s 	& 28.05s	\\
		Plan Length, $i$		
			& 7 		& 20 		& 12.80		& 2.39 			
			& 7 		& 20 		& 13.13		& 2.71  \\
		RPH, ${h}$	
			& 1 		& 5 		& 1.40		& 0.15 					
			& 1 		& 5 		& 1.27		& 0.17 \\
	\bottomrule
\end{tabular}
\label{tab:res5000overview}
\vspace*{-3mm}
\end{table}

Table~\ref{tab:res5000overview} gives an overview of the results with respect to the \majExp$\:$ experiments we performed with 7 birds for a maximum of 20 levels. The average fitness across all experiments is $0.0282$, with a standard deviation of $0.1654$.  We achieved a success rate of $94.66\%$ with fitness threshold $\varphi=10^{-3}$. The average fitness is higher than the threshold due to comparably high fitness of unsuccessful experiments. When increasing the bound for the maximal plan length $m$ to 30 we achieved a $98.4\%$ success rate in 1,000 experiments at the expense of a slightly longer average execution time.

\vspace{2.5ex}
\begin{table}
\vspace*{5mm}
	\scriptsize
	\centering
	\caption{Average duration for 100 experiments with various number of birds}
	\begin{tabular}{l cc c c c }  
		\toprule
		No. of birds & & 3 & 5 & 7 & 9 \\
      	\cmidrule(l){2-6}
        Avg. duration & & 4.58s & 18.92s & 64.85s & 269.33s \\
		\bottomrule
\end{tabular}
\label{tab:res100time}
\end{table}

The left plot in Fig.~\ref{fig:aheads_pso} depicts the resulting distribution of execution times for \majExp$\:$ runs of our algorithm, where it is clear that, excluding only a few outliers from the histogram, an arbitrary configuration of birds (Fig.~\ref{fig:form} (left)) reaches V-formation  (Fig.~\ref{fig:form} (right)) in around 1~minute. The execution time rises with the number of birds as shown in Table~\ref{tab:res100time}.

In Fig.~\ref{fig:aheads_pso}, we illustrate for how many experiments the algorithm had to increase RPH $h$ (Fig.~\ref{fig:aheads_pso} (middle)) and the number of particles used by PSO ${p}$ (Fig.~\ref{fig:aheads_pso} (right))  to improve time and space exploration, respectively.
\begin{figure}[!htb]
	\centering
    \includegraphics[width=.328\textwidth]{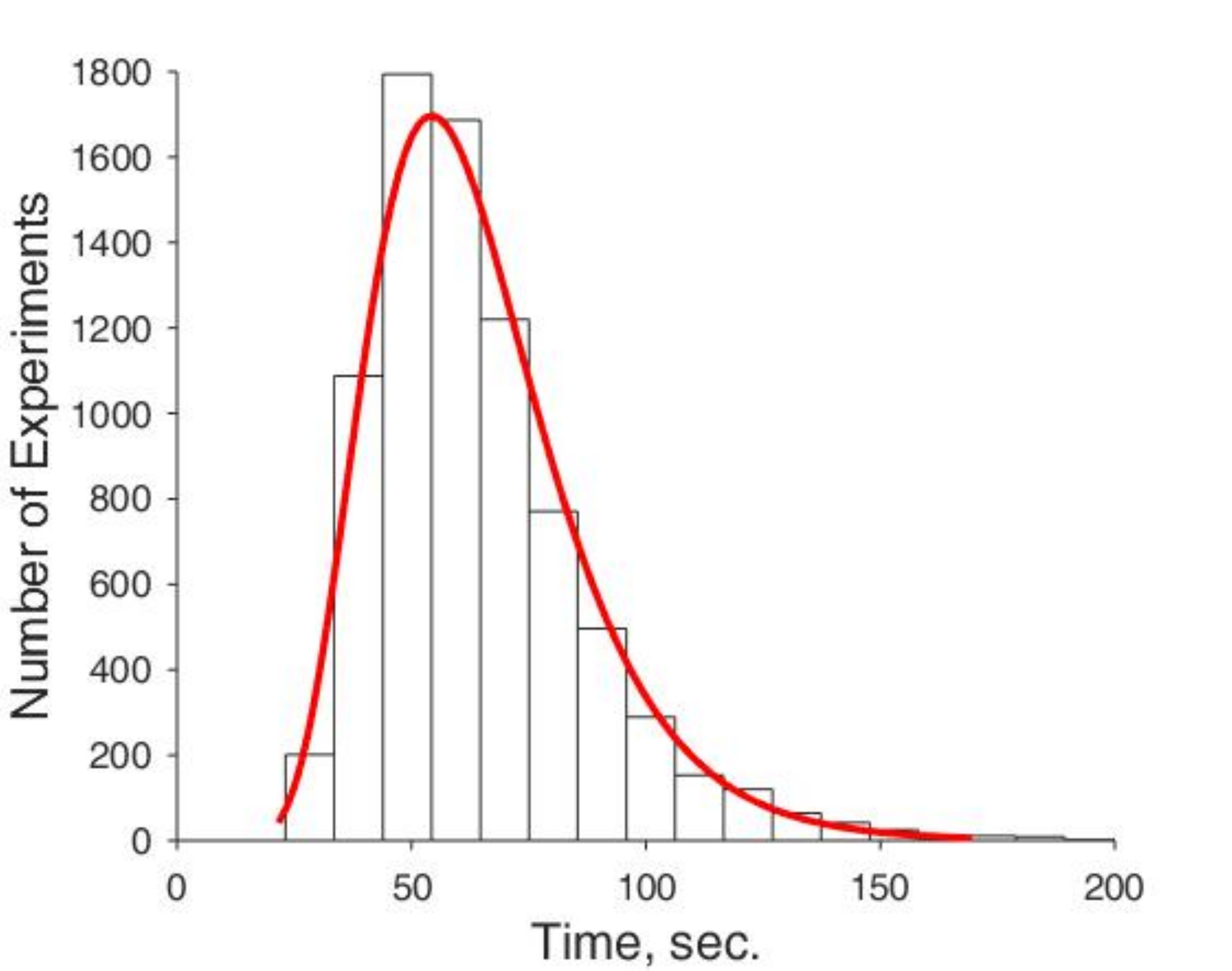}
	\includegraphics[width=.328\textwidth]{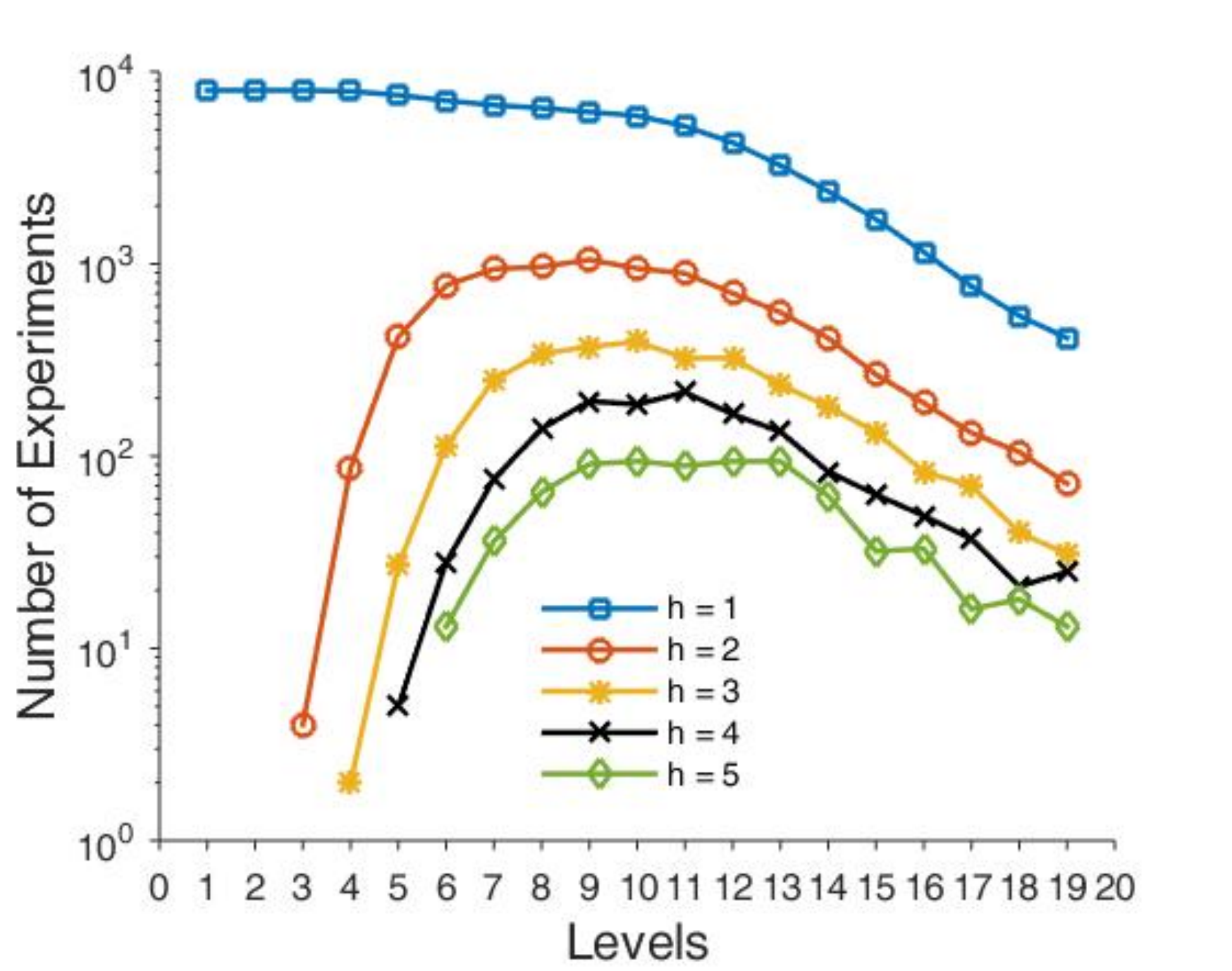}
	\includegraphics[width=.328\textwidth]{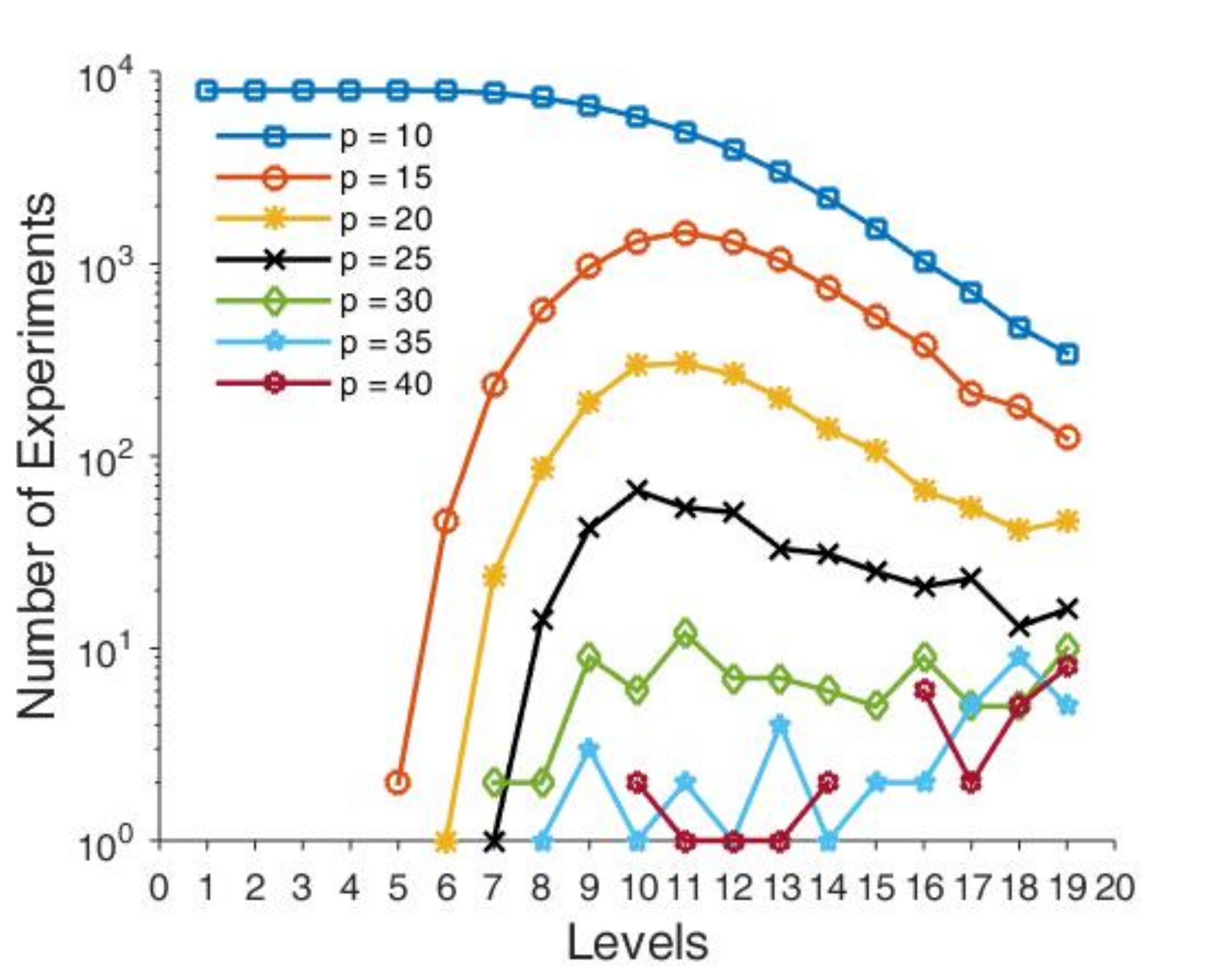}
	\caption{Left: Distribution of execution times for \majExp$\:$runs. Middle: Statistics of increasing RPH $h$. 
Right: Particles of PSO ${p}$ for \majExp$\:$experiments}\label{fig:aheads_pso}
\end{figure}

After achieving such a high success rate of ARES for an arbitrary initial configuration, we would like to demonstrate that the number of experiments performed is sufficient for high confidence in our results. This requires us to determine the appropriate number $N$ of random variables $Z_1, ... Z_N$ necessary for the Monte-Carlo approximation scheme we apply to assess efficiency of our approach. For this purpose, we use the additive approximation algorithm as discussed in~\cite{grosu2014isola}. If the sample mean $\mu_Z\,{=}\,(Z_1\,{+}\,{\ldots}\,{+}\,Z_N)/N$ is expected to be large, then one can exploit the Bernstein's inequality and fix $N$ to $\Upsilon\,{\propto}\,ln(1/\delta)/\varepsilon^2$. This results in an {\em additive} or {\em absolute-error $(\varepsilon,\delta)$-approximation scheme}:
\[
{\bf P}[\mu_Z\,{-}\,\varepsilon\leq\widetilde{\mu}_Z\leq\mu_Z\,{+}\,\varepsilon)]\geq{}1-\delta,
\]
where $\widetilde{\mu}_Z$ approximates $\mu_Z$ with absolute error $\varepsilon$ and probability $1-\delta$.

In particular, we are interested in $Z$ being a Bernoulli random variable:
 
 \[Z=\left\{
\begin{array}{ll}
1, & \text{if}\:\:J(\boldsymbol{c}(t),\va(t),{h}(t))\leqslant\varphi,\\
0, & \text{otherwise}.
\end{array}\right.\]

Therefore, we can use the Chernoff-Hoeffding instantiation of the Bernstein's inequality, and further fix the proportionality constant to $\Upsilon\,{=}\,4\,ln(2/\delta)/\varepsilon^2$, as in~\cite{HLMP04}. 

Hence, for our performed \majExp$\:$experiments, we achieve a success rate of 95\% with absolute error of $\varepsilon = 0.05$ and confidence ratio 0.99.

Moreover, considering that the average length of a plan is 13, and that each state in a plan is independent from all other plans, we can roughly consider that our above estimation generated 80,000 independent states. For the same confidence ratio of 0.99 we then obtain an approximation error $\varepsilon\,{=}\,0.016$, and for a confidence ratio of 0.999, we obtain an approximation error $\varepsilon\,{=}\,0.019$.

\section{Adaptive-Neighborhood Distributed Control}
\label{sec:dampc}

In Section~\ref{sec:ares}, we introduced the concept of Adaptive-Horizon MPC ($\ampc$). $\ampc$ gives controllers extraordinary power: we proved that under certain controllability conditions, an $\ampc$ controller can attain V-formation with probability~1. 
We now present $\dampc$~\cite{sac19}, a distributed version of $\ampc$ that extends $\ampc$ along several dimensions. First, at every time step, $\dampc$ runs a \emph{distributed consensus algorithm} to determine the optimal action (acceleration) for every agent in the flock. In particular, each agent $i$ starts by computing the optimal actions for its local subflock. The subflocks then communicate in a sequence of consensus rounds to determine the optimal actions for the entire flock.

Secondly, $\dampc$ features \emph{adaptive neighborhood resizing} in an effort to further improve the algorithm's efficiency. Like with an adaptive prediction horizon in $\ampc$, neighborhood resizing utilizes the implicit Lyapunov function to guarantee eventual convergence to a minimum neighborhood size. $\dampc$ thus treats the neighborhood size as another controllable variable that can be dynamically adjusted for efficiency purposes. This leads to reduced communication and computation compared to the centralized solution, without sacrificing statistical guarantees of convergence such as those offered by its centralized counterpart $\ampc$. Statistical global convergence can be proven.

\subsection{DAMPC System Model}
\label{subsec:dist_flock}
We consider a distributed setting with the following assumptions about the system model. 
\begin{enumerate}
\item Birds can communicate with each other without delays. As explained below, each bird $i$ adaptively changes its communication radius.
The measure of the radius is the number of birds covered, and we refer to it as bird $i$'s local neighborhood $N_i$, including bird $i$ itself.
\item All birds use the same algorithm to satisfy their local reachability goals, i.e., to bring the local cost $J\left(\vs_{N_i}\right)$, $i\,{\in}\,\{1,\ldots,B\}$, below the given threshold $\varphi$.
\item 
Birds move in continuous space and change accelerations synchronously at discrete time points.
\item After executing its local algorithms, each bird broadcasts the obtained solution to its neighbors. In this manner, every bird receives solution proposals, which differ due to the fact that each bird has its own local neighborhood. To achieve consensus, each bird takes as its best action the one with the minimal cost among the received proposals. The solutions for the birds in the considered neighborhood are then fixed. The consensus rounds repeat until all birds have fixed solutions.
\item At every time step, the value of the global cost function $J(\vs)$ is received by all birds in the flock and checked for improvement. The neighborhood for each bird is then resized based on this global check.
%
%
\item The upwash benefit $UB_i$ for bird $i$ defined in Section~\ref{subsec:dynmodel} maintains connectivity of the flock along the computations, while our algorithm manages collision avoidance.
\end{enumerate}

\subsection{The Distributed AMPC Algorithm}
\label{sec:main_algo}
In this section, we solve a stochastic reachability problem in the context of V-formation control, and demonstrate that the algorithm can be used as an alternative hill-climbing, cost-based optimization technique avoiding local minima. 

\begin{table}[htbp]\caption{Table of Notation}
\begin{center}
\small{
\begin{tabular}{r c p{.7\linewidth} }
\toprule
$H,\:h_i$ & $\triangleq$ & Maximum and current local horizon lengths\\
$N_i$ & $\triangleq$ & neighborhood of the i's bird\\
$k$ & $\triangleq$ & the number of birds in the neighborhood ($|N_i|$)\\
$m$ & $\triangleq$ & number of time-steps allowed by the property $\varphi$\\
$\va^{1\,{:}\,m}$ & $\triangleq$ & sequence of synthesized acceleration for all birds for each time-step\\
$?$ & $\triangleq$ & acceleration that has not yet been fixed\\
$1$, $!$ & $\triangleq$ & superscript for the first and last, respectively, elements in the horizon sequence\\
$\va^{1\,{:}\,!}_{N_i}$, $\vs^{1\,{:}\,!}_{N_i}$ & $\triangleq$ & sequence of accelerations and corresponding states of the horizon length reached at time-step $t$ by bird $i$ locally in its neighborhood $N_i$\\
$\Delta_i$ & $\triangleq$ & dynamical threshold defined based on the last achieved local cost $J\left(\vs^{!}_{N_j}\right)$ in the neighborhood $N_j$\\
$\va^{!},\:\vs^{!}$ & $\triangleq$ & accelerations and corresponding states for all birds achieved globally as unions of the last elements in the best horizon sequences reached locally in each neighborhood\\
$\va^{1}(t),\:\vs^{1}$ & $\triangleq$ & accelerations and states for all birds achieved globally as unions of the first elements in the best horizon sequences reached locally in each neighborhood\\
$\ell_t$ & $=$ & $J\left(\vs^{!}\right)$ -- level achieved globally at time-step $t$ after applying $\va^{1\,{:}\,!}$ to the current state\\ 
$\Delta$ & $\triangleq$ & dynamical threshold defined based on the last achieved global level\\
\bottomrule
\end{tabular}
}
\end{center}
\label{tab:notations}
\end{table}

$\dampc$ (see Alg.~\ref{alg:distributed}) takes as input an MDP $\M$, a threshold $\varphi$ defining the goal states $G$, the maximum horizon length $h_{max}$, the maximum number of time steps $m$, the number of birds $B$, and a scaling factor $\beta$. It outputs a state $\vs_0$ in $I$ and a sequence of actions $\va^{1\,{:}\,m}$ taking $\M$ from $\vs_0$ to a state in $G$.
\begin{algorithm}[ht!]
\small
	\SetKwFunction{lampc}{LocalAMPC}  
    \SetKwFunction{cost}{cost}
    \SetKwFunction{J}{J}
    \SetKwFunction{fixed}{Fixed}
    \SetKwFunction{fix}{Fix}
    \SetKwFunction{neigh}{Neighbors}
    \SetKwFunction{neighsize}{NeighSize}
    
	\SetKwInOut{Input}{Input}
	\SetKwInOut{Output}{Output}
    \SetKwFor{ParFor}{for}{do in parallel}{end} 
    \DontPrintSemicolon
	\Input{$\M\,{=}\,\left(S,A,T,J,I\right),\varphi,{h}_{\mathit{max}},m,B, \beta$}
	\Output{$\vs_0$, $\va^{1\,{:}\,m}\,{=}\,[\va(t)]_{1\leqslant t\leqslant\,m}$} 
	\BlankLine
	$\vs_0\leftarrow\mathrm{sample}(I)$; $\vs\leftarrow\vs_0$; 
    $\ell_0\leftarrow J(\vs)$; $t\leftarrow 1$; $k\leftarrow B$; $H\leftarrow h_{max}$; 
	\BlankLine
	\While{($\ell_{t{-}1} > \varphi)$ $\land$ $(t < m)$}
	{
      $\forall i:\:\va_i^{1\,{:}\,!}(t)\leftarrow\:?$; \tcp*{No bird has a fixed solution yet}
      \While 
      {\textcolor{blue!50!green}{$\left(R\leftarrow\left\{j \,|\, \va_j(t)\,{=}\,?\right\}\right)\neq\emptyset$}}
      {
       
     	\ParFor{ $i \in R$ } 
        {
          $N_i \leftarrow \neigh(i, k)$;  \:\tcp*{$k$ neighbors of $i$}
          $\Delta_i \leftarrow J\left(\vs_{N_i}^!\right)/(m{-}t)$;\;
          $\left(\vs_{N_i}^{1\,{:}\,!},\va_{N_i}^{1\,{:}\,!}\right)\leftarrow\lampc\left(\M,\vs_{N_i}^{1\,{:}\,!},\va_{N_i}^{1\,{:}\,!},\Delta_i,H,\beta\right)$;
        }  
        \textcolor{blue!50!green}{$i^* \leftarrow \argmin_{j\in R}{J\left(\vs_{N_j}^!\right)}$;\: \tcp*{Best solution in R}}
        \: \tcp{Fix $i^*$'s neighbors solutions}
        \For{$i\in\neigh(i^*,k)$}   
        {
          $\va_i^{1\,{:}\,!}(t)\leftarrow\va_{N_{i^*}}^{1\,{:}\,!}[i]$; \tcp*{The solution for bird $i$}
        }
      }
      \tcp{First action and next state}
      $\va(t)\leftarrow\va^1(t)$; $\vs^1\leftarrow\bigcup_i\vs_{N_i}^1$; $\vs^!\leftarrow\bigcup_i\vs_{N_i}^!$; $\vs\leftarrow\vs^1$;
      
      \If {\textcolor{blue!50!green}{$\ell_{t{-}1} - J\left(\vs^!\right) >\Delta$}} {
        $\ell_t\leftarrow J\left(\vs^!\right)$; $t\leftarrow t{+}1$;\tcp*{Proceed to the next level}}
      $k \leftarrow \neighsize\left(J\left(\vs^!\right), k\right)$; \tcp*{Adjust neighborhood size}
	}
	\caption{DAMPC}
	\label{alg:distributed}
\end{algorithm}
The initialization step (Line 1) chooses an initial state $\vs_0$ from $I$, fixes an initial level $\ell_0$ as the cost of $\vs_0$, sets the initial time $t$ and number of birds to process $k$. 
The outer while-loop (Lines 2-22) is active as long as $\M$ has not reached $G$ and time has not expired. In each time step, $\dampc$ first sets the sequences of accelerations $\va_i^{1\,{:}\,!}(t)$ for all $i$ to $?$ (not yet fixed), and then iterates lines~4-15 until all birds fix their accelerations through global consensus (Line 10). This happens as follows. First, all birds determine their neighborhood (\emph{subflock}) $N_i$ and the cost decrement $\Delta_i$ that will bring them to the next level (Lines 6-7). Second, they call \lampc (see Section~\ref{sec:lampc}), which takes sequences of states and actions fixed so far and extends them such that (line~8) the returned sequence of actions $\va_{N_i}^{1\,{:}\,!}$ and corresponding sequence of states $\vs_{N_i}^{1\,{:}\,!}$ decrease the cost of the subflock by $\Delta_i$. Here notation $1\,{:}\,!$ means the whole sequence including the last element $!$ (some number, the farthest point in the future where the state of the subflock is fixed), which can differ from one neighborhood to another depending on the length of used horizon. Note that an action sequence passed to $\lampc$ as input $\va_{N_i}^{1\,{:}\,!}$ contains $?$ and the goal is to fill in the gaps in solution sequence by means of this iterative process. In Line~10, we use the value of the cost function in the last resulting state $J\left(\vs^!_{N_j}\right)$ as a criterion for choosing the best action sequence proposed among neighbors $j\in R$. Then the acceleration sequences of all birds in this subflock are fixed (Lines~12-14).

After all accelerations sequences are fixed, that is, all $?$ are eliminated, the first accelerations in this sequence are selected for the output (Line~17). The next state $\vs^1$ is set to the union of $\vs_{N_i}^1$ for all neighbors $i=1\,{:}\,B$, the state of the flock after executing $\va(t)$ is set to the union of $\vs_{N_i}^!$. If we found a path that eventually decreases the cost by $\Delta$, we reached the next level, and advance time (Lines~18-20).  In that case, we optionally decrease the neighborhood, and increase it otherwise (Line~21).

The algorithm is distributed and with a dynamically changing topology. Lines~4, 10, and~18 require synchronization, which can be achieved by broadcasting corresponding information to a central hub of the network. This can be a different bird or a different base station at each time-step.

\subsection{The Local AMPC Algorithm}
\label{sec:lampc}
$\lampc$ is a modified version of the AMPC algorithm~\cite{tiwari17}, as shown in Alg.~\ref{alg:local_ampc}. 
Its input is an MDP $\M$, the current state $\vs_{N_i}^{1\,{:}\,!}$ of a subflock $N_i$, a vector of acceleration sequences $\va_{N_i}^{1\,{:}\,!}$, one sequence for each bird in the subflock, a cost decrement $\Delta_i$ to be achieved, a maximum horizon $H$ and a scaling factor $\beta$. In $\va_{N_i}^{1\,{:}\,!}$ some accelerations may not be fixed yet, that is, they have value $?$.

Its output is a vector of acceleration sequences $\va_{N_i}^{1\,{:}\,!}$, one for each bird, that decreased the cost of the flock at most, the state $\vs_{N_i}^{1\,{:}\,!}$ of the subflock after executing all actions.
\begin{algorithm}[ht!]
\small
	\SetKwFunction{PSO}{PSO}
    \SetKwFunction{cost}{cost}
	\SetKwInOut{Input}{Input}
	\SetKwInOut{Output}{Output}
    \DontPrintSemicolon
	\Input{$\M\,{=}\,(S,A,T,J,I)$, $\vs_{N_i}^{1\,{:}\,!}$, $\va_{N_i}^{1\,{:}\,!}$, $\Delta_i$, $H$, $\beta$}
	\Output{$\vs_{N_i}^{1\,{:}\,!}$, $\va_{N_i}^{1\,{:}\,!}$}
    ${p}\leftarrow{2}\cdot\beta\cdot{B}$; \tcp*{Initial swarm size}
    
    $h_i\leftarrow 1$; \tcp*{Initial horizon $\forall j\in N_i:$ $\va_{j}^1\,{=}\,?$}
	\Repeat{$\left(J\left(t\vs_{N_i}^{!}\right)-\ell_{t-1} < \Delta_i\right) \land (h_i \leqslant H)$}
	{
       \tcp{Run PSO with local information $\vs_{N_i}^{1\,{:}\,!}$ and $\va_{N_i}^{1\,{:}\,!}$}  
       $\left(t\vs_{N_i}^{1\,{:}\,!},t\va_{N_i}^{1\,{:}\,!}\right)\leftarrow$ \PSO{$\M,\vs_{N_i}^{1\,{:}\,!},\va_{N_i}^{1\,{:}\,!},p,h_i$};\\
       ${h_i}\leftarrow{h_i}+1$; 
       $p\leftarrow{2}\cdot\beta\cdot{h_i}\cdot{B}$; \tcp*{increase horizon, swarm size}
    }
    $\vs_{N_i}^{1\,{:}\,!}\leftarrow t\vs_{N_i}^{1\,{:}\,!}$; $\va_{N_i}^{1\,{:}\,!}\leftarrow t\va_{N_i}^{1\,{:}\,!}$; \tcp*{Return temporary sequences}
	\caption{LocalAMPC}
	\label{alg:local_ampc}
\end{algorithm}
$\lampc$ first initializes (Line 1) the number of particles $p$ to be used by PSO, proportionally to the input horizon $h_i$, to the number of birds $B$, and the scaling factor $\beta$. It then tries to decrement the cost of the subflock by at least $\Delta_i$, as long as the maximum horizon $H$ is not reached (Lines 3-7). 

For this purpose it calls PSO (Line 5) with an increasingly longer horizon, and an increasingly larger number of particles. The idea is that the flock might have to first overcome a cost bump, before it gets to a state where the cost decreases by at least $\Delta_i$. PSO extends the input sequences of fixed actions to the desired horizon with new actions that are most successful in decreasing the cost of the flock, and it computes from scratch the sequence of actions, for the $?$ entries. The result is returned in $\va_{N_i}^{1\,{:}\,!}$. PSO also returns the states $\vs_{N_i}^{1\,{:}\,!}$ of the flock after applying the whole sequence of actions. Using this information, it computes the actual cost achieved. 

\subsection{Dynamic Neighborhood Resizing}
\label{sec:resizing}

The key feature of $\dampc$ is that it \emph{adaptively resizes neighborhoods}. This is based on the following observation: \emph{as the agents are gradually converging towards a global optimal state, they can explore smaller neighborhoods} when computing actions that will improve upon the current configuration. 


Adaptation works on lookahead cost, which is the cost that is reachable in some future time. Line~19 of $\dampc$ is reached (and the level $\ell_t$ is incremented) whenever we are able to decrease this look-ahead cost. If level $\ell_t$ is incremented, neighborhood size $k\in [k_{min},k_{max}]$ is decremented, and incremented otherwise, as follows: $\neighsize(J, k) =$ 
\begin{align}
\begin{split}
\begin{cases}
   \min\left(\max\left(k-\left\lceil\left(1-\frac{J(s^!)}{k}\right) \right\rceil,k_{min}\right), k_{max}\right), & \text{the next level} \\  
   \min\left(k+1, k_{max}\right), & \text{otherwise}.
  \end{cases}
\end{split}\label{eq:size}
\end{align}
In Fig.~\ref{fig:example} we depict a simulation-trace example, demonstrating how levels and neighborhood size are adapting to the current value of the cost function.

\begin{figure}[t]
\centering	
\includegraphics[width=.49\linewidth]{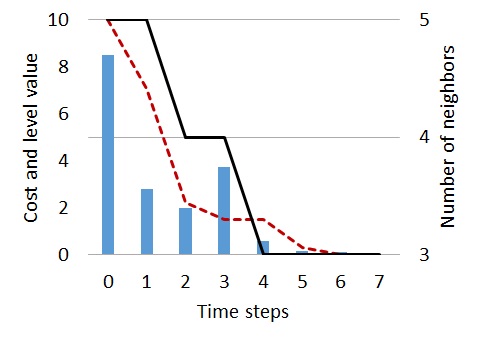}
\includegraphics[width=.49\linewidth]{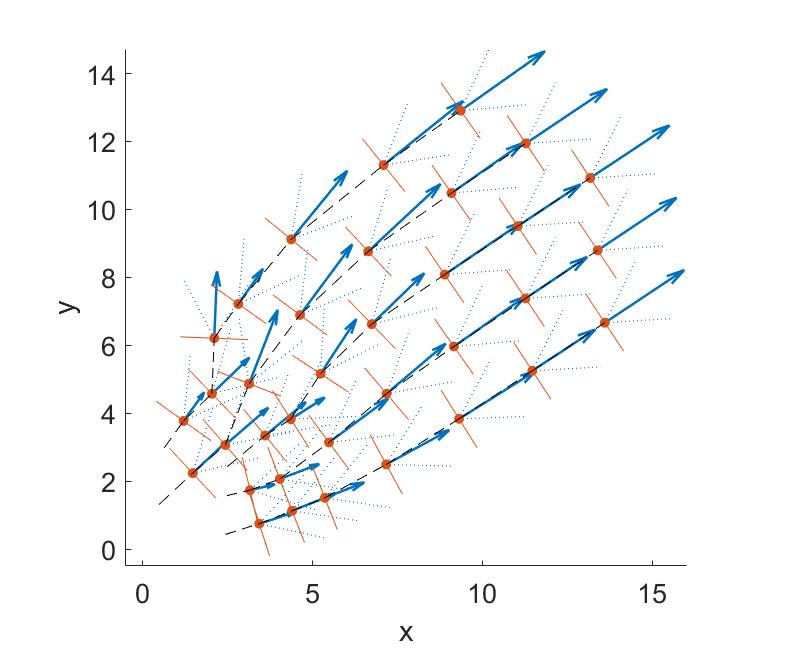}
\caption{Left: Blue bars are the values of the cost function in every time step. Red dashed line in the value of the Lyapunov function serving as a threshold for the algorithm. Black solid line is resizing of the neighborhood for the next step given the current cost. Right: Step-by-step evolution of the flock from an arbitrary initial configuration in the left lower corner towards a V-formation in the right upper corner of the plot.}
\label{fig:example}
\end{figure}

\subsection{Local Convergence}
\begin{lemma}[Local convergence]
\label{lem:local_ampc}
Given $\M=(S,A,T,J,I)$, an MDP with cost function $\texttt{cost}$, and a nonempty set of target states $G\subset S$ with $G=\{\vs\,|\,J(\vs)\leqslant\varphi\}$. If the transition relation $T$ is controllable with actions in $A$ for every (local) subset of agents, then there exists a finite (maximum) horizon $h_{max}$ such that \lampc is able to find the best actions $\va_{N_i}^{1\,{:}\,!}$ that decreases the cost of a neighborhood of agents in the states $\vs_{N_i}^{1\,{:}\,!}$ by at least a given $\Delta$.

\end{lemma}
\begin{proof}
In the input to \lampc,
the accelerations of some birds in $N_i$ may be fixed (for some horizon). As a consequence, the MDP $\M$ may not be fully controllable within this horizon. Beyond this horizon, however, PSO is allowed to freely choose the accelerations, that is, the MDP $\M$ is fully controllable again. The result now follows from convergence of AMPC (Theorem~1 from~\cite{tiwari17}).
\end{proof}

\subsection{Global Convergence and Stability}

Global convergence is achieved by our algorithm, where we overcome a local minimum by gradually adapting the neighborhood size to proceed to the next level defined by the Lyapunov function. 
\label{sec:convergence}
Since we are solving a nonlinear nonconvex optimization problem, the cost $J$ itself may not decrease monotonically. However, the look-ahead cost -- the cost of some future reachable state -- monotonically decreases. These costs are stored in level variables $\ell_t$ in Algorithm~$\dampc$ and they define a Lyapunov function $V$.
\begin{eqnarray}
V(t) = \ell_t \quad \mbox{ for levels $t = 0,1,2,\ldots$} \label{func:lyap}
\end{eqnarray}
where the levels decrease by at least a minimum dynamically defined threshold: $V(t\,{+}\,1)<V(t)-\Delta$.

\begin{lemma}
$V(t): \mathbb{Z}\rightarrow\mathbb{R}$ defined by~(\ref{func:lyap}) is a valid Lyapunov function, i.e., it is positive-definite and monotonically decreases until the system reaches its goal state. 
\end{lemma}
\begin{proof}
Note that the cost function $J(\vs)$ is positive by definition, and since $\ell_t$ equals $J(\vs)$ for some state $\vs$, $V$ is nonnegative. Line~18 of Algorithm~$\dampc$ guarantees that $V$ is monotonically decreasing by at least $\Delta$.

\end{proof}

\begin{lemma}[Global consensus]
Given Assumptions 1-7 in Section~\ref{subsec:dist_flock}, all agents in the system will fix their actions in a finite number of consensus rounds.
\end{lemma}
\begin{proof}
During the first consensus round, each agent $i$ in the system runs \lampc for its own neighborhood $N_i$ of the current size $k$. Due to Lemma~\ref{lem:local_ampc}, $\exists \widehat{h}$ such that a solution, i.e. a set of action (acceleration) sequences of length $\widehat{h}$, will be found for all agents in the considered neighborhood $N_i$. Consequently, at the end of the round the solutions for at least all the agents in $N_{i^*}$, where $i^*$ is the agent which proposed the globally best solution, will be fixed. During the next rounds the procedure recurses. Hence, the set $R$ of all agents with \texttt{nfy} values is monotonically decreasing with every consensus round.
\end{proof}

Global consensus is reached by the system during communication rounds. However, to achieve the global optimization goal we prove that the consensus value converges to the desired property.
\begin{definition}
Let $\{\vs(t)\,{:}\,t=1,2,\ldots\}$ be a sequence of random vector-variables and $s^*$ be a random or non-random. Then $\vs(t)$ \textbf{converges with probability one} to $s^*$ if \[\mathbb{P}\left[\bigcup\limits_{\varepsilon>0}\bigcap\limits_{N<\infty}\bigcup\limits_{n\geqslant N}|\vs(t)-\vs^*|\geqslant\varepsilon\right]=0.\]
\end{definition}

\begin{lemma}[Max-neighborhood convergence] If $\dampc$ is run with constant neighborhood size $B$, then it behaves identically to centralized AMPC.
\end{lemma}
\begin{proof}
If $\dampc$ uses neighborhood $B$, then it behaves like the centralized AMPC, because the accelerations of all birds are fixed in the first consensus round. 
\end{proof}

\begin{thm}[Global convergence]
Let $\M=(S,A,T,J,I)$ be an MDP with a positive and continuous cost function $J$ and a nonempty set of target states $G\,{\subset}\,S$, with $G\,{=}\,\{\vs\,|\,J(\vs)\,{\leqslant}\,\varphi\}$. If there exists a finite horizon $h_{max}$ and a finite number of execution steps $m$, such that centralized AMPC is able to find a sequence of actions $\{\va(t):\:t=1,\ldots,m\}$ that brings $\M$ from a state in $I$ to a state in $G$, then  $\dampc$ is also able to do so, with probability one. 

\label{th:global_convergence}
\end{thm}

\begin{proof}
We illustrate the proof by our example of flocking. Note that the theorem is valid in the general formulation above for the fact that as global Lyapunov function approaches zero, the local dynamical thresholds will not allow neighborhood solutions to significantly diverge from reaching the state obtained as a result of repeated consensus rounds.
Owing to Lemma~\ref{lem:local_ampc}, after the first consensus round, Alg.~\ref{alg:local_ampc} finds a sequence of best accelerations of length $h_{i^{*}}$, for birds in subflock $N_{i^{*}}$, decreasing their cost by $\Delta_{i^{*}}$. In the next consensus round, birds $j$ outside $N_{i^{*}}$ have to adjust the accelerations for their subflock $N_j$, while keeping the accelerations of the neighbors in $N_{i^{*}}\,{\cap}\,N_j$ to the already fixed solutions. If bird $j$ fails to decrease the cost of its subflock $N_j$ with at least $\Delta_j$  within prediction horizon $h_{i^{*}}$, then it can explore a longer horizon $h_j$ up to $h_{max}$. This allows PSO to compute accelerations for the birds in $N_{i^{*}}\,{\cap}\,N_j$ in horizon interval $h_j\,{<}\,h\,{\leqslant}\,h_{i^{*}}$, decreasing the cost of $N_j$ by $\Delta_j$. Hence, the entire flock decreases its cost by $\Delta$ (this defines Lyapunov function $V$ in Eq.~\ref{func:lyap}) ensuring convergence to a global optimum. If $h_{max}$ is reached before the cost of the flock was decreased by $\Delta$, the size of the neighborhood will be increased by one, and eventually it would reach $B$. Consequently, using Theorem~1 in~\cite{tiwari17}, there exists a horizon $h_{\max}$ that ensures global convergence. For this choice of $h_{max}$ and for maximum neighborhood size, the cost is guaranteed to decrease by $\Delta$, and we are bound to proceed to the next level in $\dampc$. The Lyapunov function on levels guarantees that we have no indefinite switching between ``decreasing neighborhood size'' and ``increasing neighborhood size'' phases, and we converge (see Fig.~\ref{fig:example_local}).
\end{proof}

The result presented in~\cite{tiwari17} applied to our distributed approach, together with Theorem~\ref{th:global_convergence}, ensure the following corollary.
\begin{corollary}[Global stability]
Assume the set of target states $G\in S$ has been reached and one of the following perturbations of the system dynamics has been applied: a) the best next action is chosen with probability zero (crash failure); b) an agent is displaced (sensor noise); c) an action of a player with opposing objective is performed. Then applying Algorithm~\ref{alg:distributed} the system converges with probability one from a disturbed state to a state in $G$.
\end{corollary}

\subsection{Evaluation of the Distributed $\ampc$ Controller}

We comprehensively evaluated $\dampc$ to compute statistical estimates of the success rate of reaching V-formation from an arbitrary initial state in a finite number of steps $m$.  We considered flocks of size $B=\{5,7,9\}$ birds. The specific reachability problem we addressed is as follows. 

Given a flock MDP $\M$ with $B$ birds and the randomized strategy $\sigma: S\,{\mapsto}\,\PD(A)$ of Alg.~\ref{alg:distributed}, estimate the probability of reaching a state $s$ where the cost function $J(s)\,{\leqslant}\,\varphi$, starting from an initial state in the underlying Markov chain $\M_{\sigma}$ induced by $\sigma$ on $\M$.

Since the exact solution to this stochastic reachability problem is intractable (infinite/continuous state and action spaces), we solve it approximately using statistical model checking (SMC). In particular, as the probability estimate of reaching a V-formation under our algorithm is relatively high, we can safely employ the {\em additive error} $(\varepsilon,\delta)$-Monte-Carlo-approximation scheme~\cite{grosu2014isola}. This requires $L$ i.i.d.\ executions (up to a maximum time horizon), determining in $Z_l$ if execution $l$ reaches a V-formation, and returning the mean of the random variables $Z_1,\ldots,Z_L$. 
We compute $\widetilde{\mu}_Z\,{=}\,\sum_{l=1}^LZ_l/L$ by using Bernstein's inequality to fix $L{\propto}\,ln(1/\delta)/\varepsilon^2$ and obtain
$
\mathbb{P}[\mu_Z\,{-}\,\varepsilon\leq\widetilde{\mu}_Z\leq\mu_Z\,{+}\,\varepsilon]\geq{}1\,{-}\,\delta,
$
where $\widetilde{\mu}_Z$ approximates $\mu_Z$ with additive error $\varepsilon$ and probability $1\,{-}\,\delta$. 
In particular, we are interested in a Bernoulli random variable $Z$ returning 1 if the cost $J(s)$ is less than $\varphi$ and 0 otherwise. In this case, we can use the Chernoff-Hoeffding instantiation of the Bernstein's inequality, and further fix the proportionality constant to $N\,{=}\,4\,ln(2/\delta)/\varepsilon$~\cite{herault2004approximate}. 
Executing the algorithm $10^3$ times for each flock size gives us a confidence ratio $\delta\,{=}\,0.05$ and an additive error of $\varepsilon\,{=}\,10^{-2}$.

We used the following parameters: number of birds $B\,\in\,\{5,7,9\}$, cost threshold $\varphi\,{=}\,10^{-1}$, maximum horizon $h_{max}\,{=}\,3$, number of particles in PSO $p\,{=}\,200{\cdot}h{\cdot}B$. $\dampc$ is allowed to run for a maximum of $ m\,{=}\,60$ steps. The initial configurations are generated independently, uniformly at random, subject to the following constraints on the initial positions and velocities: $\forall\:i\in\{1,\ldots,B\}\:\xv_i(0)\in[0,3]\times[0,3]$ and $\vv_i(0)\in[0.25,0.75]\times[0.25,0.75]$.

To perform the SMC evaluation of $\dampc$, and to compare it with the centralized AMPC from~\cite{tiwari17}, we designed the above experiments for both algorithms in C, and ran them on the 2x Intel Xeon E5-2660 Okto-Core, 2.2 GHz, 64 GB platform. 



Our experimental results are given in Table~\ref{tab:improv}. We used three different ways of computing the average number of neighbors for successful runs. Assuming a successful run converges after $m'$ steps, we (1)~compute the average over the first $m'$ steps, reported as ``for good runs until convergence''; (2)~extend the partial $m'$-step run into a full $m$-step run and compute the average over all $m$ steps, reported as ``for good runs over $m$ steps''; or (3)~take an average across $>m$ steps, reported as ``for good runs after convergence'', to illustrate global stability.

\begin{table}[t]
	\scriptsize
    \centering
    \caption{Comparison of DAMPC and AMPC~\cite{tiwari17} on $10^3$ runs.}
    \begin{tabular}{lccccccc}
    	\toprule
        
		& \multicolumn{3}{c}{{$\dampc$}} & 
		\multicolumn{3}{c}{{AMPC}} \\
		\cmidrule(l){2-4}\cmidrule(l){5-7}
		\textsc{Number of Birds}~~~~~~~ & {\centering\textsc{5}} & {\textsc{7}}& {\centering\textsc{9}}
		& {\textsc{5}} & {\textsc{7}} & {\textsc{9}} \\
		
        \midrule
        Success rate, $\widetilde{\mu}_Z$ &
		$0.98$ & $0.92$ & $0.80$ & $0.99$ & $0.95$ & $0.88$\\
        Avg. convergence duration, $m$ &
        $7.40$ & $10.15$ & $15.65$ & $9.01$ & $12.39$ & $17.29$\\
        Avg. horizon, $h$ &
        $1.35$ & $1.36$ & $1.53$ & $1.29$ & $1.55$ & $1.79$\\
        Avg. execution time in sec. &
        $295s$ & $974s$ & $\propto 10^3s$ & $644s$ & $3120s$ & $\propto 10^4s$\\
        
        \midrule
        Avg. neighborhood size, $k$ &&&&&&\\
        \midrule
        for good runs until convergence
        &
        $3.69$ & $5.32$ & $6.35$ & $5.00$ & $7.00$ & $9.00$\\
        for good runs over $m$ steps 
        &
        $3.35$ & $4.86$ & $5.58$ & $5.00$ & $7.00$ & $9.00$\\
        for good runs after convergence
        &
        $4.06$ & $5.79$ & $6.75$ & $5.00$ & $7.00$ & $9.00$\\
        for bad runs &
        $4.74$ & $6.43$ & $6.99$ & $5.00$ & $7.00$ & $9.00$\\
    
        \bottomrule
\end{tabular}
\label{tab:improv}
\end{table}

We obtain a high success rate for~5 and~7 birds, which does not drop significantly for 9 birds. The average convergence duration, horizon, and neighbors, respectively, increase monotonically when we consider more birds, as one would expect. The average neighborhood size is smaller than the number of birds, indicating that we improve over AMPC~\cite{tiwari17} where all birds need to be considered for synthesizing the next action. 

We also observe that the average number of neighbors for good runs until convergence is larger than the one for bad runs, except for 5~birds. 
The reason is that in some bad runs the cost drops quickly to a small value resulting in a small neighborhood size, but gets stuck in a local minimum (e.g., the flock separates into two groups) due to the limitations imposed by fixing the parameters $h_{max}$, $p$, and $m$. The neighborhood size remains small for the rest of the run leading to a smaller average.

Finally, compared to the centralized AMPC~\cite{tiwari17}, $\dampc$ is faster (e.g., two times faster for 5 birds). Our algorithm takes fewer steps to converge. The average horizon of $\dampc$ is smaller. The smaller horizon and neighborhood sizes, respectively, allow PSO to speed up its computation.

\section{Attacking the V (Controller-Attacker Games)}
\label{sec:cag}
In~~\cite{tiwari17}, we introduced \emph{V-formation games}, a class of controller-attacker games, where the goal of the controller is to maneuver the plant (a simple model of flocking  dynamics) into a V-formation,  and  the  goal  of  the  attacker is  to prevent the controller from doing so. Controllers in V-formation games use centralized $\ampc$. We define several classes of attackers, including those that in one move can remove a small number of birds from the flock, or introduce random displacement (perturbation) into the flock dynamics, again by selecting a small number of victim agents. We consider both naive attackers, whose strategies are purely probabilistic, and $\ampc$-enabled attackers, putting them  on par strategically with the controller.

We describe the specialization of the stochastic-game verification problem to
V-formation.  In particular, we present the AMPC-based control strategy for reaching a V-formation, and the various attacker strategies against which we evaluate the resilience of our controller.

\subsection{Controller's Adaptive Strategies}

Given current state $(\vec{x}(t),\vec{v}(t))$, the controller's strategy $\sigma_C$ returns a probability distribution on the space of all possible accelerations (for all birds).  As mentioned above, this probability distribution is specified implicitly via a randomized algorithm that returns an actual acceleration (again for all birds).  This randomized algorithm is the AMPC algorithm, which inherits its randomization from the randomized PSO procedure it deploys.  

When the controller computes an acceleration, it assumes that the attacker does {\em{not}} introduce any disturbances; i.e., the controller uses Eq. \ref{eq:trans} where $\va(t)$ is the only control variable. Note that the controller chooses its next action $\va(t)$ based on the current configuration $(\xv(t),\vv(t))$ of the flock using MPC. The current configuration may have been influenced by the disturbance $\vec{d}(t-1)$ introduced by the attacker in the previous time step.  Hence, the current state need not to be the state predicted by the controller when performing MPC in step $t-1$. Moreover, depending on the severity of the attacker action $\vec{d}(t-1)$, the AMPC procedure dynamically adapts its behavior, i.e.\ the choice of horizon $h$, in order to enable the controller to pick the best control action $\vec{a}(t)$ in response.

\subsection{Attacker's Strategies}
\label{subsec:games}
We are interested in evaluating the resilience of our V-formation controller when it is threatened by an attacker that can remove a certain number of birds from the flock, or manipulate a certain number of birds by taking control of their actuators (modeled by the displacement term in Eq.~\ref{eq:trans}).
We assume that the attack lasts for a limited amount of time, after which the controller attempts to bring the system back into the good set of states. When there is no attack, the system behavior is the one given by Eq.~\ref{eq:nodist}.

\vspace*{-0.5mm}\paragraph{\bf Bird Removal Game.}
In a BRG, the attacker selects a subset of R birds, where $R\,{\ll}\,B$, and removes them from the flock.  The removal of bird $i$ from the flock can be simulated in our framework by setting the displacement $\vd_i$ for bird~$i$ to $\infty$. 
We assume that the flock is in a V-formation at time $t\,{=}\,0$.  
Thus, the goal of the controller is to bring the flock back into a V-formation consisting of $B\,{-}\,R$ birds.

Apart from seeing if the controller can bring the flock back to a V-formation, we also analyze the time it takes the controller to do so. 

\begin{definition}
In a \emph{Bird Removal Game} (BRG), the attacker strategy $\sigma_D$ is defined as follows.  Starting from a V-formation of $B$ birds, i.e., $J(s_0)\leqslant \varphi$, the attacker chooses a subset of $R$ birds, $R \ll B$, by uniform sampling without replacement.  Then, in every round, it assigns each bird $i$ in the subset a displacement $\vd_i=\infty$, while for all other birds $j$, $\vd_j=0$. 

\end{definition}

\vspace*{-0.5mm}\paragraph{\bf Random Displacement Game.}
In an RDG, the attacker chooses the displacement vector for a subset of $R$ birds uniformly from the space $[0,M]\times[0,2\pi]$ with $R \ll B$. This means that the magnitude of the displacement vector is picked from the interval $[0,M]$, and the direction of the displacement vector is picked from the interval $[0,2\pi]$. We vary $M$ in our experiments. The subset of $R$ birds that are picked in different steps are not necessarily the same, as the attacker makes this choice uniformly at random at runtime as well.

The game starts from an initial V-formation. The attacker is allowed a fixed number of moves, say $20$, after which the displacement vector is identically $0$ for all birds.  The controller, which has been running in parallel with the attacker, is then tasked with moving the flock back to a V-formation, if necessary. 
\begin{definition}
In a \emph{Random Displacement Game} (RDG), the attacker strategy $\sigma_D$ is defined as follows.  Starting from a V-formation of $B$ birds, i.e., $J(s_0)\leqslant \varphi$, in every round, it chooses a subset of $R$ birds, $R \ll B$, by uniform sampling without replacement. It then assigns each bird $i$ in the subset a displacement $\vd_i$ chosen uniformly at random from $[0,M]\times[0,2\pi]$, while for all other birds $j$, $\vd_j=0$. After $T$ rounds, all displacements are set to $0$.
\end{definition}

\vspace*{-0.5mm}\paragraph{\bf{AMPC Game.}}
\begin{figure}[t]
	\centering
	\includegraphics[width=0.80\linewidth]{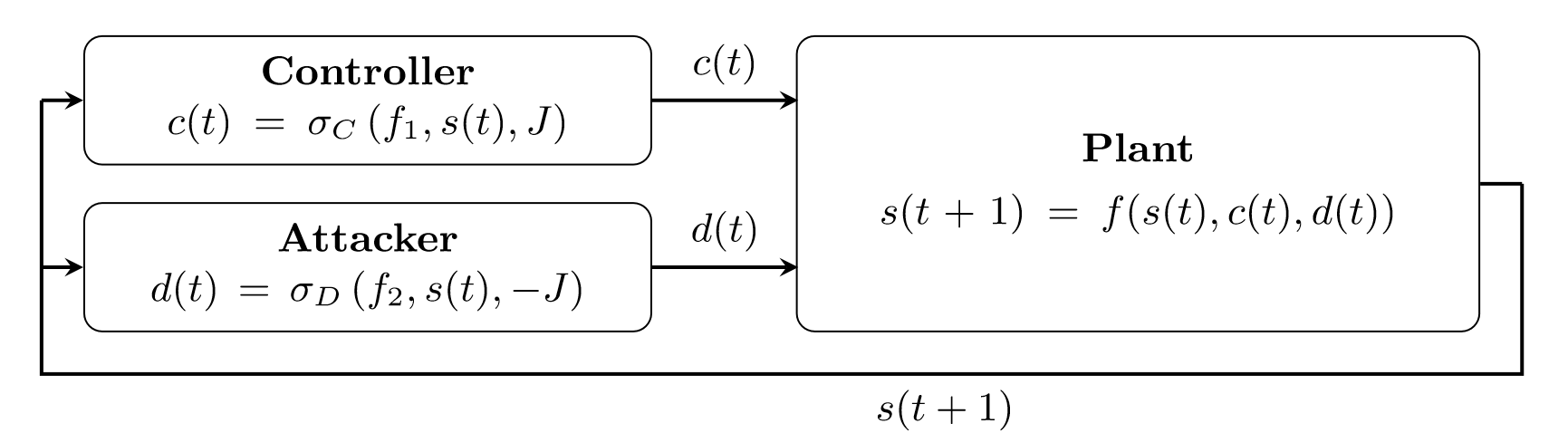}
\vspace*{-2mm}
\caption{Controller-Attacker Game Architecture. The controller and the attacker use randomized strategies $\sigma_C$ and $\sigma_D$ to choose actions $c(t)$ and $d(t)$ based on dynamics, respectively, where $s(t)$ is the state at time $t$, and $f$ is the dynamics of the plant model. The controller tries to minimize the cost $J$, while the attacker tries to maximize it.}
\label{fig:ampc}
     \vspace*{-3mm}
\end{figure}
An AMPC game is similar to an RDG except that the attacker does not use a uniform distribution to determine the displacement vector. The attacker is advanced and strategically calculates the displacement using the AMPC procedure. See Figure~\ref{fig:ampc}.  In detail, the attacker applies AMPC, but assumes the controller applies zero acceleration. Thus, the attacker uses the Eq. \ref{eq:trans} as the model of the flock dynamics.

Note that the attacker is still allowed to have $\vd_i(t)$ be non-zero for only a small number of birds. However, it gets to choose these birds in each step.  It uses the AMPC procedure to simultaneously pick the subset of $R$ birds and their displacements. The objective of the attacker's AMPC is to maximize the cost.

\begin{definition}
In an \emph{AMPC game}, the attacker strategy $\sigma_D$ is defined as follows. Starting from a V-formation of $B$ birds, i.e., $J(s_0)\leqslant \varphi$, in every round, it uses AMPC to choose a subset of  $R$ birds, $R \ll B$, and their displacements $\vd_i$ for bird~$i$ in the subset from $[0,M]\times[0,2\pi]$; for all other birds $j$, $\vd_j=0$. After $T$ rounds, all displacements are set to $0$.
\end{definition}

\begin{thm}[AMPC Convergence]
\label{thm:ampc}
Given an MDP $\M\,{=}\,(S,A,T,J)$ with positive and continuous cost function $J$, and a nonempty set of target states $G\,{\subset}\,S$ with $G\,{=}\,\{s\,|\,J(s)\,{\leqslant}\,\varphi\}$. If the transition relation $T$ is controllable with actions in $A$, then there exists a finite maximum horizon $h_{\mathit{max}}$ and a finite number of execution steps $m$, such that AMPC is able to find a sequence of actions $a_1,\ldots,a_m$ that brings a state in $S$ to a state in $G$ with probability one.
\end{thm}

\begin{proof}
In each (macro-) step of horizon length $h$, from level $\ell_{i-1}$ to level $\ell_i$, AMPC decreases the distance to $\varphi$ by $\Delta_i\,{\geqslant}\,\Delta$, where $\Delta\,{>}\,0$ is fixed by the number of steps $m$ chosen in advance. Hence, AMPC converges to a state in $G$ in a finite number of steps, for a properly chosen $m$. AMPC is able to decrease the cost in a macro step by $\Delta_i$ by the controllability assumption and the fairness assumption about the PSO algorithm. Since AMPC is a randomized algorithm, the result is probabilistic. Note that the theorem is an existence theorem of $h_{\mathit{max}}$ and $m$ whose values are chosen empirically in practice.
\end{proof}

\begin{thm}[AMPC resilience in a C-A game]
\label{thm:resilience}
Given a controller-attacker game, there exists a finite maximum horizon $h_{\mathit{max}}$ and a finite maximum number of game-execution steps $m$ such that AMPC controller will win the controller-attacker game in $m$ steps with probability~1.
\end{thm}

\begin{proof}
Since the flock MDP (defined by Eq.~\ref{eq:trans}) is controllable, the PSO algorithm we use is fair, and the attack has a bounded duration, the proof of the theorem follows from Theorem~\ref{thm:ampc}. 
\end{proof}

\begin{remark}
While Theorem~\ref{thm:resilience} states that the controller is expected to win with probability~1, we expect winning probability to be possibly lower than one in many cases because: (1)~the maximum horizon $h_{\mathit{max}}$ is fixed in advance, and so is (2)~the maximum number of execution steps $m$; (3)~the underlying PSO algorithm is also run with bounded number of particles and time. Theorem~\ref{thm:resilience} is an existence theorem of $h_{\mathit{max}}$ and $m$, while in practice one chooses fixed values of $h_{\mathit{max}}$ and $m$ that could be lower than the required values.
\end{remark}

\subsection{Statistical MC Evaluation of V-Formation Games}
\renewcommand{\majExp}{{2,000}}

The stochastic-game verification problem we address in the context of the V-formation-AMPC algorithm is formulated as follows.  Given a flock MDP $\M$ (we consider the case of $B\,{=}\,7$ birds), acceleration actions $\va$ of the controller, displacement actions $\vd$ of the attacker, the randomized strategy $\sigma_C: S\,{\mapsto}\,\PD(C)$ of the controller (the AMPC algorithm), and a randomized strategy
$\sigma_D: S\,{\mapsto}\,\PD(D)$ for the attacker,
determine the probability of reaching a state $s$ where the cost function $J(s)\,{\leqslant}\,\varphi$ (V-formation in a 7-bird flock), starting from an initial state (in this case this is a V-formation), in the underlying Markov chain induced by strategies $\sigma_C$, $\sigma_D$ on $\M$.

Since the exact solution to this reachability problem is intractable due to the infinite/continuous space of states and actions, we solve it approximately with classical statistical model-checking (SMC).  The particular SMC procedure we use is from~\cite{grosu2014isola} and based on an {\em additive} or {\em absolute-error $(\varepsilon,\delta)$-Monte-Carlo-approximation scheme}.

This technique requires running $N$ i.i.d.\ game executions, each for a given maximum time horizon, determining if these executions reach a V-formation, and returning the average number of times this occurs. 

Each of the games described in Section~\ref{subsec:games} is executed 2,000 times. For a confidence ratio $\delta\,{=}\,0.01$, we thus obtain an additive error of $\varepsilon\,{=}\, 0.1$.
We use the following parameters in the game executions: number of birds $B\,{=}\,7$, threshold on the cost $\varphi\,{=}\,10^{-3}$, maximum horizon $h_{\mathit{max}}\,{=}\,5$, number of particles in PSO $p\,{=}\,20hB$. In BRG, the controller is allowed to run for a maximum of $30$ steps. In RDG and AMPC game, the attacker and the controller run in parallel for $20$ steps, after which the displacement becomes $0$, and the controller has a maximum of $20$ more steps to restore the flock to a V-formation.

To perform SMC evaluation of our AMPC approach we designed the above experiments in C and ran them on the Intel Core i7-5820K CPU with 3.30 GHz and with 32GB RAM available.

\begin{figure}[t]
 \centering
  \includegraphics[width=.495\textwidth]{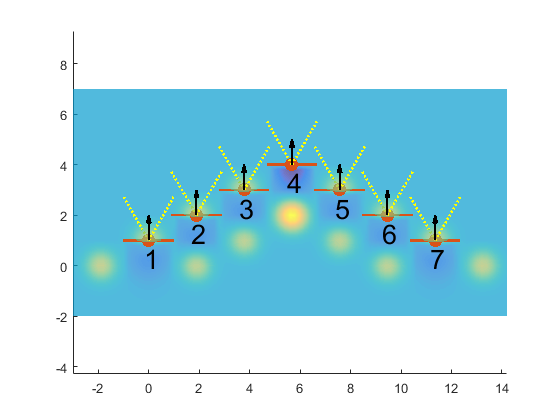}
  \includegraphics[width=.495\textwidth]{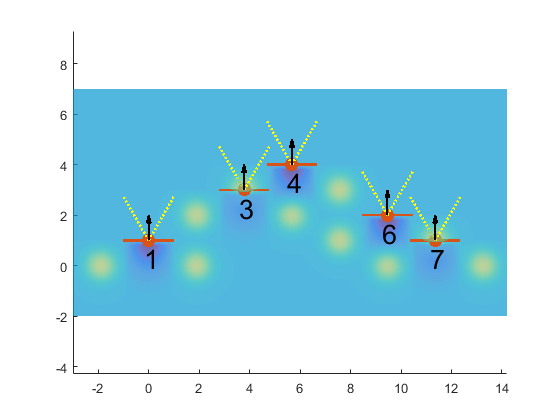}
  \caption{Left: numbering of the birds. Right: configuration after removing Bird 2 and 5. The red-filled circle and two protruding line segments represent a bird's body and wings. Arrows represent bird velocities. Dotted lines illustrate clear-view cones. A brighter/darker background color indicates a higher upwash/downwash.}
  \label{fig:numbering}
  \vspace{3mm}
 \end{figure}

\begin{table}[ht]
\centering
\caption {Results of 2,000 game executions for removing 1 bird with $h_{\mathit{max}}\,{=}\,5$,  $m\,{=}\,40$}
    \begin{tabular}{lcccccccccc}
    \toprule
		&&& Ctrl. success rate, \% &&& Avg. convergence duration &&& Avg. horizon\\ \midrule
    Bird 4  && & $99.9$           && & $12.75$  &&& $3.64$                            \\
    Bird 3  && & $99.8$           && & $18.98$  &&& $4.25$                            \\
    Bird 2  && & $100$           && & $10.82$  &&& $3.45$                            \\ \bottomrule
    \end{tabular}
\label{tab:resRemoveOne}
\end{table}

\begin{table}[ht]
	\centering
	\caption{Results of 2,000 game executions for removing 2 birds with $h_{\mathit{max}} \,{=}\,5$, $m\,{=}\,30$}
	\begin{tabular}{lccccccccccc}
		\toprule
		&&& Ctrl. success rate, \% &&& Avg. convergence duration &&& Avg. horizon\\
      	\midrule
        Birds 2 and 3 & & &$0.8$ & & &$25.18$&&&$4.30$\\
        Birds 2 and 4 & & &$83.1$ & & &$11.11$ &&&$2.94$\\
        Birds 2 and 5 & & &$80.3$ & & &$9.59$ &&&$2.83$\\
        Birds 2 and 6 & & &$98.6$ & & &$7.02$ &&&$2.27$\\
        Birds 3 and 4 & & &$2.0$ & & &$22.86$ &&&$4.30$\\
        Birds 3 and 5 & & &$92.8$ & & &$11.8$ &&&$3.43$\\
		\bottomrule
\end{tabular}
\label{tab:resRemoveTwo}
\end{table}

\begin{table}[t]
	\centering
	\caption{Results of 2,000 game executions for random displacement and AMPC attacks with $h_{\mathit{max}}\,{=}\,5$ and $m\,{=}\,40$ (attacker runs for 20 steps)}
	\begin{tabular}{cccccccccccc}  
		\toprule
		Range of noise &&& Ctrl. success rate, \% &&& Avg. convergence duration &&& Avg. horizon\\
      	\midrule
        &&&\multicolumn{7}{c}{Random displacement game}\\
        \cmidrule(l){3-10}
        $[0,0.50]\times[0,2\pi]$ & & &$99.9$ & & &$3.33$ &&&$1.07$\\
        $[0,0.75]\times[0,2\pi]$ & & &$97.9$ & & &$3.61$ &&&$1.11$\\
        $[0,1.00]\times[0,2\pi]$ & & &$92.3$ & & &$4.14$ &&&$1.18$\\
 		\cmidrule(l){3-10}
        &&&\multicolumn{7}{c}{AMPC game}\\
        \cmidrule(l){3-10}
        $[0,0.50]\times[0,2\pi]$  && & $97.5$    &&& $4.29$ &&& $1.09$\\
        $[0,0.75]\times[0,2\pi]$  && & $63.4$    &&& $5.17$ &&& $1.23$\\
        $[0,1.00]\times[0,2\pi]$  && & $20.0$    &&& $7.30$ &&& $1.47$\\
		\bottomrule
\end{tabular}
\label{tab:resRandomNoise}
\end{table}

\vspace*{-0.5mm}\paragraph{\bf Discussion of the Results}
To demonstrate the resilience of our adaptive controller, for each game introduced in Section~\ref{subsec:games}, we performed a number of experiments to estimate the probability of the controller winning.  Moreover, for the runs where the controller wins, the average number of steps required by the controller to bring the flock to a V-formation is reported as {\em{average convergence duration}}, and the average length of the horizon used by AMPC is reported as {\em{average horizon}}.

The numbering of the birds in Tables~\ref{tab:resRemoveOne} and~\ref{tab:resRemoveTwo} is given in Figure~\ref{fig:numbering}. Bird-removal scenarios that are symmetric with the ones in the tables are omitted. The results presented in Table~\ref{tab:resRemoveOne} are for the BRG game with $R\,{=}\,1$.
In this case, the controller is {\em{almost always}} able to bring the flock back to a V-formation, as is evident from Table~\ref{tab:resRemoveOne}. Note that removing Bird $1$ (or $7$) is a trivial case that results in a V-formation.

In the case when $R\,{=}\,2$, shown in Table~\ref{tab:resRemoveTwo}, the success rate of the controller depends on {\em{which two birds are removed}}. Naturally, there are cases where dropping two birds does not break the V-formation; for example, after dropping Birds~1 and~2, the remaining birds continue to be in a V-formation.  Such trivial cases are not shown in Table~\ref{tab:resRemoveTwo}. Note that the scenario of removing Bird~$1$ (or~$7$) and one other bird can be viewed as removing one bird in flock of $6$ birds, thus not considered in this table. Among the other nontrivial cases, the success rate of controller drops slightly in four cases, and drops drastically in remaining two cases. 
This suggests that attacker of a CPS system can incur more damage by being prudent in the choice of the attack. 

Impressively, whenever the controller wins, the controller needs about the same number of steps to get back to V-formation (as in the one-bird removal case). On average, removal of two birds results in a configuration that has worse cost compared to an BRG with $R\,{=}\,1$. 
Hence, the adaptive controller is able to make bigger improvements (in each step) when challenged by worse configurations. Furthermore, among the four cases where the controller win rate is high, experimental results demonstrate that removing two birds positioned asymmetrically with respect to the leader poses a stronger, however, still manageable threat to the formation. For instance, the scenarios of removing birds~2 and~6 or 3 and~5 give the controller a significantly higher chance to recover from the attack, $98.6\%$ and $92.8\%$, respectively.

Table~\ref{tab:resRandomNoise} explores the effect of making the attacker smarter. Compared to an attacker that makes random changes in displacement, an attacker that uses AMPC to pick its action is able to win more often. This again shows that an attacker of a CPS system can improve its chances by cleverly choosing the attack. For example, the probability of success for the controller to recover drops from $92.3\%$ to $20.0\%$ when the attacker uses AMPC to pick displacements with magnitude in $[0,1]$ and direction in $[0,2\pi]$. The entries in the other two columns in Table~\ref{tab:resRandomNoise} reveal two even more interesting facts.

First, in the cases when the controller wins, we clearly see that the controller uses a longer look-ahead when facing a more challenging attack. This follows from the observation that the average horizon value increases with the strength of attack. This gives evidence for the fact that the adaptive component of our AMPC plays a pivotal role in providing resilience against sophisticated attacks.
Second, the average horizon still being in the range $1$-$1.5$, means that the adaptation in our AMPC procedure also helps it perform better than a fixed-horizon MPC procedure, where usually the horizon is fixed to $h\,{\geqslant}\,2$.
When a low value of $h$ (say $h\,{=}\,1$) suffices, the AMPC procedure avoids unnecessary calculation that using a fixed $h$ might incur.

In the cases where success rate was low (Row~1 and Row~5 in Table~\ref{tab:resRemoveTwo}, and Row~3 of the AMPC game in Table~\ref{tab:resRandomNoise}), we conducted additional 500 runs for each case and observed improved success rates ($2.4\%$, $9\%$ and $30.8\%$ respectively) when we increased $h_{\mathit{max}}$ to~$10$ and $m$ to~$40$. This shows that success rates of AMPC improves when given more resources, as predicted by Theorem~\ref{thm:ampc}.


\section{Related Work}
\label{sec:related}


Organized flight in flocks of birds can be categorized in \emph{cluster 
flocking} and \emph{line formation}~\cite{heppner1974avian}. In cluster 
flocking the individual birds in a large flock seem to be uncoordinated in general. However, the flock moves, turns, and wheels as if it were one organism. In 1987 Reynolds~\cite{Reynolds1987CG} defined his three famous rules describing separation, alignment, and cohesion for individual birds in order to have them flock together. This work has been great inspiration for research in the area of collective behavior and self-organization.

In contrast, line formation flight requires the individual birds to fly in a very specific formation. Line formation has two main benefits for the 
long-distance migrating birds. First, exploiting the generated uplift by birds flying in front, trailing birds are able to conserve 
energy~\cite{lissaman1970formation,Cutts251,weimerskirch2001nature}. Second, in a staggered formation, all birds have a clear view in front as well as a view on their neighbors~\cite{Bajec2009AB}. While there has been quite some effort to keep a certain formation for multiple entities when traveling 
together~\cite{Seiler2002CDC, Gennaro2005CAIC,Dang2015CYBCONF}, only little 
work deals with a task of achieving this extremely important formation from a random starting configuration \cite{Cattivelli2011TSP}. The convergence of bird flocking into V-formation has been also analyzed with the use of combinatorial techniques\cite{Chazelle:2014}. 

Compared to previous work, in~\cite{MPC2007} this question is addressed without using any behavioral rules but as problem of \emph{optimal control}. In~\cite{yang2016love} a cost function was proposed that reflects all major features of V-formation, namely, \emph{Clear View} (CV), \emph{Velocity Matching} (VM), and \emph{Upwash Benefit} (UB). The technique of \gls{mpc} is used to achieve V-formation starting from an arbitrary initial configuration of $n$ birds. \gls{mpc} solves the task by minimizing a functional defined as squared distance from the optimal values of CV, VM, and UB, subject to constraints on input and output. The approach is to choose an optimal \emph{velocity adjustment}, as a control input, at each time-step applied to the velocity of each bird by predicting model behavior several time-steps ahead. 
 
The controller synthesis problem has been widely studied~\cite{Plans2012}.
The most popular and natural technique is Dynamic Programming (DP)~\cite{Bellman:1957}, which improves the approximation of the functional at each iteration, eventually converging to the optimal one given a fixed asymptotic error. Compared to DP, which considers all possible states of the system and might suffer from state-space explosion in case of environmental uncertainties, approximate algorithms~\cite{Henriques2012,Bartocci2016,mannor_cross_2003,bartlett_experiments_2011,stulp_policy_2012,stulp_path_2012} take into account only the paths leading to a desired target. One of the most efficient ones is Particle Swarm Optimization (PSO)~\cite{Kennedy95particleswarm} that has been adopted for finding the next best step of MPC in~\cite{yang2016love}. Although it is a very powerful optimization technique, it has not yet been possible to achieve a high success rate in solving the considered flocking problem.

Sequential Monte-Carlo methods prove to be efficient in tackling the question of control for linear stochastic systems~\cite{chen2009fast}, in particular, Importance Splitting (IS)~\cite{KalajdzicIsola16}. The approach we propose is, however, the first attempt to combine adaptive IS, PSO, and receding-horizon technique for \emph{synthesis of optimal plans for controllable systems}. 
We use MPC to synthesize a plan, but use IS to determine the intermediate fitness-based waypoints.  We use PSO to solve the multi-step optimization problem generated by MPC,  but choose the planning horizon and the number of particles adaptively. These choices are governed by the difficulty to reach the next level.


Adaptive control, and its special case of adaptive model predictive control, typically refers to the aspect of the controller updating its process model that it uses to compute the control action. The field of adaptive control is concerned with the discrepancy between the actual process and its model used by the controller. In our adaptive-horizon MPC, we adapt the lookahead horizon employed by the MPC, and not the model itself.  Hence, the work in this paper is orthogonal to what is done in adaptive control~\cite{adaptive_control,adaptive_mpc}.

Adaptive-horizon MPC was used in~\cite{droge2011adaptive}
to track a reference signal. If the reference signal is unknown, and we have a poor estimate of its future behavior, then a larger horizon for MPC is not beneficial. Thus, the horizon was determined by the uncertainty in the knowledge of the future reference signal. We consider cost-based reachability goals here, which allows us to choose a horizon in a more generic way based on the progress toward the goal. More recently, adaptive horizons were also used in~\cite{krener2016adaptive} for a reachability goal. However, they chose a large-enough horizon that enabled the system to reach states from where a pre-computed local controller could guarantee reachability of the goal. This is less practical than our approach for establishing the horizon.

Prior work on the V-formation problem has focused on giving combinations of \emph{dynamical flight rules} as driving forces. These approaches tend to be distributed in nature, as flight rules describe how an individual bird maneuvers depending on the positions and velocities of the neighbors within its radius of influence.  
For instance, in~\cite{flake1998computational}, the authors extend Reynolds' flocking model~\cite{reynolds1987flocks} with a rule that forces a bird to move laterally away from any bird that blocks its view. This can result in multiple V-shaped clusters, but flock-wide convergence is not guaranteed. The work of~\cite{dimock2003aerodynamic} induces V-formations by extending Reynolds' model with a \emph{drag reduction} rule,
but the final formation tends to oscillate as birds repeatedly adjust the angle of the~V. Another approach, based on three \emph{positioning rules}, is that of~\cite{nathan2008}. It provides an alternative model that produces V-formations. The birds in their model follow three positioning rules: (1) seek the proximity of the nearest bird; (2) seek the nearest position that affords an unobstructed longitudinal view; and (3) attempt to position itself in the upwash of a leading bird.
Their model, however, is limited by the assumption that the birds have a constant longitudinal heading. The authors of~\cite{stonedahl2011finding} attempt to improve upon this approach by handling \emph{turning movements}. This also forms small clusters of birds, each of which is only moderately V-like.


In~\cite{ZhanL13}, the problem of taking an arbitrary initial configuration of $n$ agents to a final configuration where every pair of ``neighbors'' is a fixed distance $d$ apart, and every agent is stationary (its velocity is zero) is considered. They present centralized and distributed algorithms for this problem, both of which use MPC to determine the next action. The problem addressed in~\cite{ZhanL13} is arguably simpler than the V-formation problem we consider. The cost function being minimized in their case is a quadratic convex function. In~\cite{ZhanL13} the proof of convergence uses the fact of existing sequence of states with monotonically decreasing cost. Our cost function is nonconvex and nonlinear, which requires overcoming local minima by horizon and neighborhood adaptation. Both of these concepts are not required, and hence not addressed, in~\cite{ZhanL13}. In the distributed control procedure of~\cite{ZhanL13}, each agent publishes the control value it locally computed, which is then used by other agents to compute their control value. A quadratic number of these ``small steps'' are performed before each agent fixes its control input for the next time step. Our distributed procedure has at most a linear number of these small steps.

Other related work, including~\cite{Fowler02,Dandrea03,Ye17}, focuses on distributed controllers for flight formation that operate in an environment where the (multi-agent) plant is already in the desired formation and the (distributed) controller's objective is to maintain formation in the presence of disturbances (typically on the roll angle of the wing). Moreover, the plants considered are more physically detailed than our plant model in terms of capturing the dynamics of moving-wing aircraft. We plan to consider models of this nature as future work. 

A distinguishing feature of these approaches is the particular formation they are seeking to maintain, including a half-vee~\cite{Fowler02}, a ring and a torus~\cite{Dandrea03}, and a leader-follower formation~\cite{Ye17}.  In contrast, we use distributed AMPC with dynamic neighborhood resizing to bring a flock from a mostly random initial configuration to a stable V-formation.


In the field of CPS security, one of the most widely studied attacks is \emph{sensor spoofing}.  When sensors measurements are compromised, state estimation becomes challenging, which inspired a considerable amount of work on attack-resilient state estimation~\cite{DBLP:journals/tac/FawziTD14,Bullo13:TAC,Pajic14:ICCPS,Pajic15:ICCPS,UAVspoofing}.
In these approaches, resilience to attacks is typically achieved by assuming the presence of redundant sensors, or coding sensor outputs.
In our work, we do not consider sensor-spoofing attacks, but assume the attacker gets control of the displacement vectors (for some of the birds/drones).  We have not explicitly stated the mechanism by which an attacker obtains this capability, but it is easy to envision ways (radio controller, attack via physical medium, or other channels~\cite{savage}) for doing so.


A key focus in CPS security has also been detection of attacks. For example, recent work considers displacement-based attacks on formation flight~\cite{aiaa2016}, but it primarily concerned with detecting which UAV was attacked using an unknown-input-observer based approach. We are not concerned with detecting attacks, but establishing that the adaptive nature of our controller provides attack-resilience for free. Moreover, in our setting, for both the attacker the and controller the state of the plant is completely observable.
In~\cite{saulnier2017resilient}, a control policy based on the robustness of the connectivity graph is proposed to achieve consensus on the velocity among a team of mobile robots, in the present of non-cooperative robots that communicate false values but execute the agreed upon commands. In contrast, we allow the attacker to manipulate the executed commands of the robots. The cost function we use is also more flexible so that we can encode more complicated objectives.

We are unaware of any work that uses statistical model checking to evaluate the resilience of adaptive controllers against (certain classes of) attacks.

\section{Conclusions}
\label{sec:concl}
We first presented ARES, a very general adaptive, receding-horizon synthesis algorithm for MDP-based optimal plans; ARES can be viewed as a model-predictive controller with an adaptive receding horizon (AMPC).  We conducted a thorough performance analysis of ARES on the V-formation problem to obtain statistical guarantees of convergence. For flocks of 7 birds, ARES is able to generate, with high confidence, an optimal plan leading to V-formation in 95\% of the 8,000 random initial configurations we considered, with an average execution time of only 63 seconds per plan.

We next presented $\dampc$, a distributed version of AMPC that uses an adaptive-neighborhood and adaptive-horizon model-predictive control algorithm to generated actions for a controllable MDP so that it eventually reaches a state with cost close to zero, provided that the MDP has such a state. The main contribution of $\dampc$ as a distributed control algorithm is that it adaptively resizes an agent's local neighborhood, while still managing to converge to a goal state with high probability. Our evaluation showed that the value of dynamic neighborhood resizing, where we observed that it can lead to a relatively small average neighborhood size while successfully reaching a goal state.

Finally, to demonstrate the resilience of our adaptive controllers, we introduced a variety of controller-attacker games and carried out a number of experiments to estimate the probability of the controller winning.  Our analysis demonstrated the effectiveness of adaptive controllers in overcoming certain kinds of controller-targeted attacks.

\section*{Acknowledgment}
The authors gratefully acknowledge the significant contributions of Ezio Bartocci, Lukas Esterle, Christian Hirsch, and Junxing Yang to this work.


\bibliographystyle{is-alpha}
\newcommand{\etalchar}[1]{$^{#1}$}


  
\end{document}